\documentclass[12pt,letterpaper]{article}
\usepackage[T1]{fontenc}
\usepackage{lmodern}
\usepackage[utf8]{inputenc}
\usepackage{geometry}
\usepackage{amsmath}
\usepackage{amsthm} 
\usepackage{amsfonts}
\usepackage{dsfont}
\usepackage{cancel}
\usepackage{ amssymb }
\usepackage{bm}
\usepackage{enumitem}
\usepackage{mathtools}
\usepackage{changepage}
\usepackage{amsthm}
\usepackage{verbatim}
\usepackage{amsmath,amssymb}
\usepackage{graphicx}
\usepackage{emptypage}
\usepackage{newlfont}
\usepackage[all,cmtip]{xy}
\usepackage[bottom]{footmisc}
\usepackage[titletoc,title]{appendix}
\usepackage{chngcntr}
\usepackage{apptools}
\usepackage{listings}
\usepackage{color}
\usepackage{sgame, tikz} 
\usepackage{kpfonts}  
\usepackage{natbib}
\AtAppendix{\counterwithin{lem}{section}}
\usepackage{multibib}
\newcites{app}{References}
\usepackage{caption}
\usepackage{subcaption}
\usepackage{float}
\usepackage{thmtools,thm-restate}
\usepackage{epigraph}
\usepackage{xr-hyper}
\usepackage{datetime}
\usepackage{sgamevar}
\usepackage{setspace}
\usepackage[colorlinks=true,linkcolor=blue, allcolors=blue]{hyperref}

\theoremstyle{plain} 
\newtheorem{cor}{Corollary} 
\newtheorem{prop}{Proposition}
\newtheorem{claim}{Claim} 
\newtheorem{theorem}{Theorem}
\newtheorem{lemma}{Lemma}
\theoremstyle{definition} 
\newtheorem{ex}{Example}

\theoremstyle{remark} 
\newtheorem{rmk}{Remark} 

\let\emptyset\varnothing

\DeclareMathOperator{\E}{\mathds{E}}
\renewcommand{\P}{\mathds{P}}

\newcommand{\R}{\mathds{R}}

\newcommand{\N}{\mathds{N}}

\newcommand{\M}{\mathcal{M}}

\newcommand{\Y}{\mathcal{Y}}
\newcommand{\IC}{\text{IC}}
\newcommand{\IR}{\text{IR}}

\renewcommand{\1}{\mathds{1}}

\newcommand\eqid{\stackrel{d}{=}}

\renewcommand*\d{\mathop{}\!\mathrm{d}}

\newcommand{\X}{\mathcal{X}}

\theoremstyle{definition}

\reversemarginpar
\allowdisplaybreaks

\newdate{draftdate}{28}{01}{2025}
\usdate

\usepackage[nameinlink]{cleveref}
\crefname{manualasm}{assumption}{assumptions}
\crefname{cor}{corollary}{corollaries}
\crefalias{prop}{proposition}
\crefname{claim}{claim}{claims}
\crefname{ex}{example}{examples}
\crefname{defn}{definition}{definitions}
\crefname{rmk}{remark}{remarks}

\begin{document}

\title{Costly Multidimensional Screening\thanks{I thank the coeditor, Andrea Galeotti, and the five anonymous referees for very thoughtful comments. I would also like to thank Mohammad Akbarpour, Gabriel Carroll, Daniel Chen, Piotr Dworczak, Mira Frick, Soheil Ghili, Nima Haghpanah, Jason Hartline, Andreas Haupt, Ravi Jagadeesan, Zi Yang Kang, Yingkai Li, Paul Milgrom, Mike Ostrovsky, Anne-Katrin Roesler, Ilya Segal, Eran Shmaya, Ludvig Sinander, Andy Skrzypacz, Takuo Sugaya, Bob Wilson, Ali Yurukoglu, Weijie Zhong, and several seminar and conference audiences for their helpful comments and suggestions.}
}
\author{Frank Yang\thanks{Department of Economics, Harvard University. Email: \href{mailto:fyang@fas.harvard.edu}{fyang@fas.harvard.edu}.}}
\date{\today}
\maketitle
\begin{abstract}
A screening instrument is \textit{\textbf{costly}} if it is socially wasteful and \textit{\textbf{productive}} otherwise. A principal screens an agent with multidimensional private information and quasilinear preferences that are additively separable across two components: a one-dimensional productive component and a multidimensional costly component. Can the principal improve upon simple one-dimensional mechanisms by also using the costly instruments? We show that if the agent has preferences between the two components that are positively correlated in a suitably defined sense, then simply screening the productive component is optimal. The result holds for general type and allocation spaces, and allows for nonlinear and interdependent valuations. We discuss applications to monopoly pricing, bundling, and labor market screening. 
\\

\end{abstract}
\bigskip
\setcounter{page}{1}

\newpage
\section{Introduction}

Actions convey information. The effort to obtain credentials conveys information about the ability of a job applicant. The time spent waiting in line conveys information about the willingness to pay of a customer. The endurance of physical activity conveys information about the health status of an individual.\footnote{The New York Times reports, ``The [Wal-Mart] memo suggests that the company could require all jobs to include some component of physical activity, like making cashiers gather shopping carts.'' \textit{Wal-Mart's health care struggle is corporate America's, too}, The New York Times, October 29, 2005. See also \citet{zeckhauser2021strategic}, who argues that socially wasteful ordeals play a prominent role in health care.} These actions are often \textit{\textbf{costly}} in that they are socially wasteful.  However, because the preferences over these actions are correlated in some way with the private information that affects the allocation of productive assets, the informational content from these costly actions---in addition to the usual price instrument---could be useful for screening.

The addition of a nonprice screening instrument significantly complicates the screening problem, by making the allocation space multidimensional. Multidimensional mechanisms are often far more powerful than one-dimensional mechanisms because they can intricately link the incentive constraints (\citealt{Rochet2003}). For example, consider the following parable of \citet{Stiglitz2002}: An insurance company ``might realize that by locating itself on the fifth floor of a walk-up building, only those with a strong heart would apply. [...] More subtly, it might recognize that how far up it needs to locate itself depends on other elements of the strategy such as premium charged.''  

In this paper, we study the effectiveness of costly nonprice screening. Under what conditions should we expect these costly instruments to be used in the design of optimal contracts? Does assuming away such nonprice screening always lead to a suboptimal mechanism in this richer space of mechanisms? To address these questions, we put forward a new multidimensional screening model. The model consists of two components: \textit{(i)} a productive component which the principal intrinsically cares about (such as insurance coverage), and \textit{(ii)} a costly component which the principal may utilize to help with screening but destroys social surplus (such as walking up stairs). 

In the model, the principal designs a mechanism to assign the productive allocations in a one-dimensional space $\X$ and the costly actions in an arbitrary space $\Y$. Monetary transfers are allowed. Both the principal and the agent have quasilinear preferences that are additively separable across the two components  $\X$ and $\Y$. We say that the agent's preferences are \textit{\textbf{positively correlated}} between the two components if the type who has higher utility for the productive allocations tends to also have lower disutility for the costly actions in the stochastic dominance sense. 

Our main result (\Cref{thm:main}) states that if the agent has preferences between the two components that are positively correlated, then simply screening the one-dimensional productive component is optimal (and essentially uniquely optimal). We also provide a partial converse (\Cref{prop:converse}) showing that for a given negative correlation structure, there exist utility functions such that the optimal contract must involve costly screening. We allow the agent to have multidimensional private information; however, our result is new even when the agent has one-dimensional types because of the multidimensionality of the screening instruments (i.e. price and nonprice). 

A basic intuition behind our result can be understood as follows. Consider the parable of \citet{Stiglitz2002} with two types: a high-risk type and a low-risk type. Suppose that the high-risk type is fully insured and the low-risk type is less than fully insured. Now, suppose that the insurance company can make the contracts contingent on a physical activity such as climbing the stairs. If the high-risk type is less fit, then the company could increase the coverage targeted at the low-risk type. To purchase this contract, the individual would have to climb the stairs. Since the high-risk type finds it harder to climb the stairs than the low-risk type does, such a contract could be incentive compatible and increase the profit. On the other hand, suppose that the task were to wait in line in order to be eligible for the contract. If the high-risk type is more likely to be unemployed and hence finds waiting less costly, then the company would not benefit from this instrument, since it would make it easier for the high-risk type to mimic the low-risk type. 

However, this intuition is incomplete because it assumes monotonicity of the allocation rule, which is \textit{with} loss of generality in multidimensional settings.\footnote{Implementability in multidimensional environments is characterized by \textit{\textbf{cyclic monotonicity}} (see \citealt{rochet1987necessary}) which allows a much richer set of allocation rules.} In particular, in the above example, the firm can be better off by not insuring the high-risk type. In the standard setting, this is impossible since this allocation is not monotone in the type. But suppose that the firm decreases the coverage targeted at the high-risk type and its associated price but requires a long waiting time. Because the low-risk type finds it more costly to wait, such a non-monotone contract involving a nonprice instrument can be incentive compatible and even optimal (see \Cref{rmk:surplus}). 

A difficulty with proving our result is that the allocation rules in our setting need not be monotone. 
However, we show that quite generally the simple intuition turns out to lead to the right prediction.
The proof deals with the richer space of multidimensional mechanisms. It relies on two key ingredients. First, we show that if the agent has positively correlated preferences between the two components, then costly instruments can only help with upward incentive constraints. In particular, for a given multidimensional mechanism, we first \textit{decompose} the multidimensional type space into a collection of one-dimensional paths, and then show that on each path, we can \textit{reconstruct} a one-dimensional mechanism that involves no costly screening, improves on the original mechanism, and satisfies all downward incentive constraints. Second, we show that only downward incentive constraints are needed in any one-dimensional screening model that satisfies a single-crossing condition on the surplus function (the \textit{\textbf{surplus condition}}). This second ingredient, which we call the \textit{\textbf{downward sufficiency theorem}} (\Cref{thm:dbind}), also uncovers a novel property of one-dimensional screening models. 

In the insurance example, our result implies that (see \Cref{subsec:csignal}) costly instruments can be useful if the higher-risk type has higher costs (e.g. climbing the stairs) and cannot be useful if the higher-risk type has lower costs (e.g. waiting in line). The same result also applies to the classic setting of job market screening. Positive correlation of preferences arises there because a higher ability applicant often tends to find both the work easier to accomplish and education less costly. Our result implies that a monopsonistic firm need not make its offers contingent on the costly signals from an applicant, despite the fact that the firm prefers a higher ability applicant (see \Cref{subsec:labor}). Thus, complementary to the classic case of competitive wages \`{a} la   \citet{Spence1973Job}---which leaves very little room to screen types via monetary payments, and hence all the screening has to occur via costly efforts---we show that costly screening may not be needed in a monopsonistic market.\footnote{Indeed, \Cref{app:comp} studies a version of our model with competitive firms and shows that, consistent with \citet{Spence1973Job}, costly screening can occur in equilibrium under positive correlation of preferences.}

Beyond these direct implications, our result turns out to also yield new insights into the classic multiproduct pricing problem that may at first glance appear to be unrelated (see \Cref{subsec:bundling}). In the multiproduct pricing problem, a monopolist wants to jointly price different bundles, each of which generates a positive surplus. Our key insight is that selling the bundle of all goods can be viewed as the productive component, and selling smaller bundles \textit{instead of} the grand bundle can be viewed as the costly instruments for screening values of the grand bundle. Using this perspective, as an application, we recover a recent result by \citet{haghpanah2021pure} on when \textit{\textbf{pure bundling}} (i.e. selling only the grand bundle) is optimal.\footnote{As we show, the same logic also allows us to generalize their result to a multiple-good monopoly problem allowing for both
probabilistic bundling and quality discrimination; the optimal mechanism there generally involves price discrimination but does so only along the quality dimension (see \Cref{rmk:generalratio}).} More generally, it turns out that we can also view any menu of nested bundles as a productive component, and any bundle not in the nested menu as a costly instrument. Applying our main result, we also immediately obtain new conditions for when \textit{\textbf{nested bundling}} (i.e. selling a nested menu of bundles) is optimal, complementing recent studies of nested bundling (e.g. \citealt{yang2023nested}). 

\subsection{Related Literature}
\label{sec:lit}

This paper introduces a mechanism design framework studying price and nonprice screening. Multidimensional screening differs significantly from its one-dimensional counterpart and remains elusive to fully characterize. Much of the literature focuses on the multiple-good monopoly problem. When there is a single good, the optimal mechanism is simply a posted price (\citealt{Myerson1981}; \citealt{riley1983optimal}). However, as soon as there is more than one good, seemingly simple special cases (such as two goods with additive and independent values) turn out to be analytically intractable (\citealt{thanassoulis2004haggling}; \citealt{pycia2006stochastic}; \citealt{manelli2006bundling}). Recent literature has been focused on identifying tractable special cases (\citealt{mcafee1988multidimensional}; \citealt{pavlov2011property}; \citealt*{daskalakis2017strong}; \citealt{haghpanah2021pure}; \citealt{bikhchandani2022selling}; \citealt{ghili2023characterization}; \citealt{yang2023nested}), 
and studying approximately optimal mechanisms (\citealt{babaioff2014simple}; \citealt*{cai2016duality}; \citealt{hart2017approximate}) as well as worst-case optimal mechanisms (\citealt{Carroll2017}; \citealt{che2021robustly}; \citealt{deb2023}).

As in \citet{yang2023nested}, this paper also identifies a new tractable case of multidimensional screening. Conceptually, it offers a new perspective by considering a multidimensional screening model in which all dimensions except one are surplus destructive. The model turns out to yield clear economic insights into when multidimensional mechanisms can or cannot improve on one-dimensional mechanisms. The multiproduct pricing problem can be viewed as a special case of our model by redefining the allocation space. Among this literature, the closest paper to ours is \citet{haghpanah2021pure}, who study the optimality of pure bundling. Building on \citet{haghpanah2021pure}, our positive correlation condition uses the notion of stochastic monotonicity and relies on the \textit{canonical representation} technique in dynamic mechanism design to decompose the type space (\citealt{Eso2007}; \citealt*{Pavan2014}). Our key technical contribution is the downward sufficiency theorem, which allows us to solve the screening problem with multiple instruments after decomposing the type space. The downward sufficiency theorem relies on a single-crossing condition on the surplus function, which is also new and automatically satisfied in \citet{haghpanah2021pure}.

Our model of costly screening builds on the literature on mechanism design with money burning (\citealt{banerjee1997theory}; \citealt{hartline2008optimal}; \citealt{condorelli2012money}). In particular,  \citet{amador2013theory,amador2020money} and \citet{ambrus2017delegation} allow for money burning in a delegation model and characterize when the optimal mechanism does not rely on money burning. A recent line of work (\citealt{Condorelli2013}; \citealt*{ortoleva2022cares}; \citealt*{akbarpour2024redistributive}) compares non-market mechanisms with market mechanisms when the agent's willingness to pay is informative about their welfare weight. Since welfare weights affect only the principal's objective, these models have a one-dimensional effective type, and the relevant correlation is about the agent's and principal's preferences.  By contrast, in our model, the multiple dimensions of private information are all payoff-relevant to the agent, and the correlation condition is about the agent's preferences between different screening instruments.\footnote{In a follow-up work, \citet*{yang2024comparison} study the comparison of different costly screening devices when monetary transfers are not allowed.}

The remainder of the paper proceeds as follows. \Cref{sec:model} presents our model. \Cref{sec:main} presents the main result and a partial converse. \Cref{sec:proof} presents the proof of the main result. \Cref{sec:app} presents the applications. \Cref{sec:conclusion} concludes. \Cref{app:proof} provides the omitted proofs. \Cref{app:b} (online appendix) provides additional results and proofs. 

\section{Model}\label{sec:model}
A principal wants to screen an agent. The agent has private information summarized by a multidimensional type $\theta = (\theta^A, \theta^B)$, where $\theta^A \in \Theta^A \subseteq \R$ and $\theta^B \in \Theta^B \subseteq \R^N$ for a finite $N$; for convenience, sometimes we also refer to $\theta^A$ as $\theta^0$ and $\theta^B$ as $(\theta^1, \dots, \theta^N)$.  We use the superscripts $A$, $B$ to indicate the productive and costly components, respectively. 

Both $\Theta^A$ and $\Theta^B$ are assumed to be compact. Let $\Theta := \Theta^A \times \Theta^B$ denote the type space; let $\Delta(\Theta)$ denote the space of Borel probability measures on $\Theta$, equipped with the weak-$^*$ topology. The agent's type is drawn from a commonly known distribution $\mu\in \Delta(\Theta)$. 

The space of \textit{\textbf{productive allocations}} $\X \ni x$ is a compact subset of $\R$; the space of \textit{\textbf{costly instruments}} $\Y \ni y$ is an arbitrary measurable space. 

Both the principal and the agent have quasilinear preferences that are additively separable across the two components: The principal's (ex post) payoff is given by 
\[v^A(x, \theta^A) + v^B(y, \theta^B) + t\,, \]
and the agent's payoff is given by 
\[u^A(x, \theta^A) + u^B(y, \theta^B) - t\,, \]
where $t$ denotes \textit{\textbf{monetary transfer}}. The utility functions for the productive component $u^A$, $v^A$ are assumed to be continuous on $\X \times \Theta^A$; those for the costly component $u^B$, $v^B$ are allowed to be any bounded measurable functions on $\Y \times \Theta^B$. The principal has \textit{\textbf{interdependent preferences}} if $v^A$ or $v^B$ depends on the agent's type. 

The (ex post) surplus functions for the two components are denoted by 
\[s^A(x, \theta^A)  := u^A(x, \theta^A) + v^A(x, \theta^A)\,,\quad s^B(y, \theta^B)  := u^B(y, \theta^B) + v^B(y, \theta^B)  \,. \]
The defining feature of the costly component is that any allocation is socially wasteful under complete information: for all $y \in \Y$ and all $\theta^B \in \Theta^B$, 
\[\label{eq:star} s^B(y, \theta^B) \leq 0  \,. \tag{1}\]
We assume that there is an element $y_0 \in \Y$ representing \textit{\textbf{no costly screening}}: 
\[v^B(y_0, \theta^B) = u^B(y_0, \theta^B) = 0 \,. \tag{2}\]
We say the instruments are \textit{\textbf{strictly costly}} if \eqref{eq:star} holds strictly for all $y \neq y_0$ and all $\theta^B$.  

The game proceeds as follows. The principal posts a menu of options $\big\{(x, y, t)\big\}$. The agent either declines, which we assume results in payoff $0$ for both parties, or selects an option $(x, y, t)$ from the menu, pays monetary transfer $t$, takes costly action $y$, and then gets allocation $x$. By the revelation principle, it is without loss of generality to restrict attention to direct mechanisms. A \textit{\textbf{(direct, incentive-compatible) mechanism}} is a measurable map
\[(x, y, t): \Theta \rightarrow \X \times \Y \times \R\]
satisfying the usual incentive compatibility (IC) and individual rationality (IR) constraints:
\begin{alignat*}{2}
u^A(x(\theta), \theta^A) + u^B(y(\theta), \theta^B) - t(\theta) &\geq u^A(x(\hat{\theta}), \theta^A) + u^B(y(\hat{\theta}), \theta^B) - t(\hat{\theta})  \quad &&\text{for all } \theta, \hat{\theta} \in \Theta  \,;  \\ 
u^A(x(\theta), \theta^A) + u^B(y(\theta), \theta^B) - t(\theta) &\geq 0 &&\text{for all }  \theta \in \Theta  \,.
\end{alignat*}
Let $\M(\Theta)$ denote the space of IC and IR mechanisms. The principal wants to solve 
\[\sup_{(x, y, t) \in \M(\Theta)} \E[v^A(x(\theta), \theta^A) + v^B(y(\theta), \theta^B) + t(\theta)]  \,.\]
A mechanism $(x, y, t)$ \textit{\textbf{involves no costly screening}} if $y(\theta) = y_0$ for all $\theta$ and $(x, t)$ does not depend on $\theta^B$, in which case the mechanism screens only the productive component. A mechanism $(x, y, t)$ \textit{\textbf{involves no costly screening almost everywhere}} if $P(y(\theta) = y_0) = 1$. 

We impose the classic increasing differences assumption on the productive component: the agent's utility $u^A(x, \theta^A)$ is nondecreasing in $\theta^A$, and for any $x < \hat{x}$, $\theta^A < \hat{\theta}^A$,  
\[u^A(\hat{x}, \theta^A)  - u^A(x, \theta^A)  < u^A(\hat{x}, \hat{\theta}^A)  - u^A(x, \hat{\theta}^A)\,.\]
We impose a new condition on the surplus function $s^A(x, \theta^A)$: for any $x < \hat{x}$,  $\theta^A<\hat{\theta}^A$, 
    \[ \hspace{41mm} s^A(\hat{x}, \theta^A) > s^A(x, \theta^A) \implies s^A(\hat{x},  \hat{\theta}^A)  > s^A(x,  \hat{\theta}^A)\,.  \hspace{3mm}  \textbf{\textit{(Surplus Condition)}}\]
This condition says that if a higher allocation $x$ generates more surplus at some type $\theta^A$, then it continues to do so for higher types.

To state our notion of positive correlation of preferences between the two components, we introduce some notation. Let $\preceq_{st}$ denote the usual stochastic order for $\R^N$-valued random variables, i.e. $X \preceq_{st} Y$ if $\E[f(X)] \leq \E[f(Y)]$ for all bounded nondecreasing (measurable) functions $f:\R^N \rightarrow \R$.\footnote{We say a function $f:\R^N \rightarrow \R$ is \textit{\textbf{nondecreasing}} if $f(x) \leq f(y)$ for all $x \leq y \in \R^N$, where $\leq$ is the componentwise weak inequality.} Let $\theta^B \mid \theta^A$ denote the regular conditional distribution of $\theta^B$ given $\theta^A$.\footnote{See e.g. \citeauthor{klenke2013probability} (\citeyear{klenke2013probability}, pp. 180-185). } We assume that the agent's utility $u^B(y, \theta^B)$ is nondecreasing in $\theta^B$ so that $\theta^B$ represents the strength of the preferences on the costly component. We say that the agent has \textit{\textbf{positively correlated preferences (between the two components)}} if for any $\theta^A < \hat{\theta}^A $, 
\[ \hspace{60mm} \theta^B \mid \theta^A  \preceq_{st} \theta^B \mid \hat{\theta}^A\,. \hspace{30mm} \textbf{\textit{(Positive Correlation)}}\]
This condition says that an agent who has a higher $\theta^A$ tends to have a higher $\theta^B$ in the stochastic dominance sense. 

\subsection{Discussion of Assumptions}

\paragraph{Costly Instruments.}\hspace{-2mm}A key assumption in our model is that costly instruments do not generate surplus. This assumption makes it clear that the costly instruments are only useful as screening devices. We impose very little preference structure on the costly component: Both the allocation space $\mathcal{Y}$ and the type space $\Theta^B$ are allowed to be multidimensional; the utility functions $u^B$ and $v^B$ are also fully general. The generality helps highlight the key driving force of our main result, and can also be useful in applications (as we discuss below). It also clarifies that for our main result, we need the agent's preferences to exhibit the positive correlation property for all costly instruments available.\footnote{As we show in the partial converse, if there were a costly instrument that displays negative correlation, then the principal could be strictly better off by using costly screening.}

\paragraph{Additive Separability.}\hspace{-2mm}The model also assumes that the preferences are additively separable across the two components. Additive separability is an important assumption as it allows us to collapse the multidimensional screening problem into a one-dimensional screening problem when costly screening is not used.\footnote{Because costly screening would never be used under surplus-maximizing allocations, such dimension collapsing also implies that ex post efficient allocation rules are generally implementable in our setting in contrast to \citet{jehiel2001efficient} where the ex post efficient allocations depend nontrivially on the multidimensional types and are generally impossible to implement.}  
However, in certain applications with non-additive preferences (such as the bundling application in \Cref{subsec:bundling}), because the surplus-maximizing allocations depend only on a one-dimensional type, we may be able to transform the problem into this model by considering an additive preferences structure where the costly component is multidimensional (i.e. encoding all allocations that cannot be surplus-maximizing for any type).

\paragraph{Surplus Condition.}\hspace{-2mm}Our surplus condition is a single-crossing condition on the surplus function in the productive component. It is a new condition but commonly satisfied in one-dimensional screening problems.\footnote{By \textit{\textbf{one-dimensional screening problems}}, we refer to the ones that satisfy the increasing differences condition; we maintain this terminology throughout.} For example, this condition automatically holds when the principal's preferences are not interdependent, given increasing differences in the agent's preferences. In general, however, this condition differs from the increasing differences condition on the agent's preferences. Sufficient conditions for the surplus condition include, for example, \textit{(i)} $s^A$ is strictly increasing in $x$, or \textit{(ii)} $s^A$ is twice-differentiable with the cross partial derivative $s^A_{12} \geq 0$. This assumption ensures that there is a monotone efficient allocation rule. It is not satisfied when the principal's preference to trade with low types is so strong that any socially efficient allocation rule is not monotone. This is an important assumption for our key technical result, the downward sufficiency theorem (\Cref{thm:dbind}).  As we discuss in \Cref{rmk:surplus}, in the absence of this condition, costly screening can be profitable even if the positive correlation condition holds because the principal may profitably use it to deter upward deviations.

\paragraph{Positive Correlation.}\hspace{-2mm}The positive correlation condition is also known as stochastic monotonicity (\citealt{muller2002comparison}). In particular, we will equivalently say that $\theta^B$ is \textit{\textbf{stochastically nondecreasing}} in $\theta^A$ whenever the positive correlation condition holds. This is an asymmetric condition. It says that observing a high $\theta^A$ conveys good news about $\theta^B$ in the sense of stochastic dominance. A sufficient but not necessary condition is that $(\theta^A, \theta^B)$ are affiliated in the sense of \citet{Milgrom1982} (see \Cref{lem:aff} in \Cref{app:add}). Note that this statistical condition can encode the agent's preferences between the two components because we assume that both $u^A(x,\,\cdot\,)$ and $u^B(y, \,\cdot\,)$ are nondecreasing in $\theta^A$ and $\theta^B$, respectively.\footnote{A closely related stochastic monotonicity condition is used by \citet{haghpanah2021pure} for the optimality of pure bundling---their condition asks the ratios of any bundle values relative to the grand bundle values to be stochastically monotone in the grand bundle values. As we explain in \Cref{subsec:bundling}, their condition can be viewed as a positive correlation of preferences between a productive component and a costly component in an appropriate sense.}

\section{Main Result} \label{sec:main}
Our main result says that if the agent has positively correlated preferences between the productive and costly components, then simply screening the one-dimensional productive component is optimal and essentially uniquely optimal: 

\begin{restatable}{theorem}{thm:main}\label{thm:main} Under~the~surplus~condition, if~the~agent~has~positively~correlated~preferences,~then: 
\begin{itemize}
    \item[(i)] There exists an optimal mechanism that involves no costly screening. 
    \item[(ii)] If the instruments are strictly costly, then every optimal mechanism involves no costly screening almost everywhere. 
\end{itemize}
\end{restatable}

In the case of negatively correlated preferences, we show a partial converse:

\begin{prop}[Partial converse]\label{prop:converse}
Suppose that the type distribution has a density, $|\X|>1$, and $|\Y|>1$. Suppose that there exists some dimension $i$ in the costly component such that $\theta^i$ is stochastically \emph{nonincreasing} in $\theta^A$ and they are not independent. Then, there exist utility functions ($u^A$, $u^B$, $v^A$, $v^B$) satisfying the surplus condition such that any mechanism screening only the productive component is strictly dominated by a mechanism involving costly screening. 
\end{prop}

We discuss the intuition behind \Cref{thm:main} in \Cref{subsec:intuition} and \Cref{subsec:intuition_proof}, and then present the proof in \Cref{sec:proof}. \Cref{prop:converse} can be shown by a simple construction that sets the principal's utility functions $v^A$ and $v^B$ to be $0$.\footnote{In a follow-up work, \citet{yang2023nested} uses the connection between bundling and costly screening to obtain necessary and sufficient conditions for costly screening to be optimal in the case of perfectly negatively correlated preferences.} The proof proceeds as follows. The principal can always create a menu of two nontrivial options for the agent: \textit{(i)} getting the favorite allocation in $\X$ at a high price, and \textit{(ii)} getting the same allocation at a low price but with some costly activity. The proof shows that if $\theta^i$ is negatively correlated with $\theta^A$ as defined in the statement, then there exist some utility functions for the agent such that this way of price discrimination is always more profitable for the principal than selling the elements in $\X$ alone. The appendix provides details.

\subsection{Discussion of Intuition with Two Types}\label{subsec:intuition}
In this section, we present the key intuition behind our main result. We start by presenting two binary-type examples, building on the discussion in the Introduction. We then discuss how the intuition generalizes to the case of multiple and potentially multidimensional types. 

The first example shows how the principal can make use of the costly instrument if the agent has negatively correlated preferences between the two components. 
\begin{ex}[Negative correlation]\label{ex:negcor}
A principal sells insurance to an agent who has two types: $\theta^A = 0$ is the low-risk type and $\theta^A = 1$ is the high-risk type. Suppose the types have equal probabilities. The low-risk type has value $2$ for the insurance, whereas the high-risk type has value $3$. It costs the insurance firm $0$ to serve the low-risk type, and $\frac{5}{2}$ to serve the high-risk type. Then, the agent's and the principal's utilities for the productive component are: 
\[u^A(x,\theta^A) = (\theta^A + 2) x, \qquad v^A(x, \theta^A) = - \frac{5}{2} \theta^A x\,, \]
where $x \in [0, 1]$ represents insurance coverage. Suppose that the costly instrument $y$ is \textit{climbing the stairs}, as in the Introduction, for which the low-risk type has cost $0$ and the high-risk type has cost $1$. Then, the agent's and the principal's utilities for the costly component are: 
\[\qquad \qquad \qquad \quad  \qquad u^B(y,\theta^B) = \theta^B y, \qquad v^B(y, \theta^B) = 0 \,, \quad \text{where $\theta^B = - \theta^A$ }\,.\]
The agent has negatively correlated preferences here because the high-risk type has a
higher utility for the insurance but also higher disutility for the costly action. 

If the principal does not use the costly instrument, then this becomes a one-dimensional screening problem, and one may verify that the optimal one-dimensional mechanism is to offer full insurance (i.e., $x = 1$) at the price of $2$ to sell to both types. This yields a profit 
\[2 \times 1 - \frac{1}{2} \times \frac{5}{2} = \frac{3}{4}\,.\]
However, the principal can mitigate the adverse selection problem by only offering full insurance to the agent if the agent pays the price $2$ and also climbs the stairs. The high-risk type, finding climbing the stairs too costly, does not purchase this plan. Then, the principal gets a profit 
\[
\pushQED{\qed} 
2 \times \frac{1}{2} = 1 > \frac{3}{4}\,. \qedhere
\popQED\] 
\end{ex}

The next example shows that, when the agent has positively correlated preferences,  even though there are incentive-compatible mechanisms in which the costly instrument is used, such a mechanism is dominated by one without costly screening.

\begin{ex}[Positive correlation]\label{ex:poscor}
The productive component is the same as in \Cref{ex:negcor}. We change the costly component. Suppose that the costly instrument $y$ is \textit{waiting in line} for which the low-risk type ($\theta^A = 0$) has cost $1$ and the high-risk type ($\theta^A = 1$) has cost $0$. Then, the agent's and the principal's utilities for the costly component are:
\[\qquad \qquad \qquad  u^B(y,\theta^B) = (\theta^B-1) y, \qquad v^B(y, \theta^B) = 0 \,, \quad \text{where $\theta^B = \theta^A$ }\,.\]
The agent has positively correlated preferences here because the high-risk type has a higher utility for the insurance but also lower disutility for the costly action.

There are incentive-compatible mechanisms in which the costly instrument is used. For example, suppose that the principal offers two options:
\begin{itemize}
    \item[\textit{(i)}] a full-insurance plan with a price $2$
    \item[\textit{(ii)}] a low-coverage plan with coverage $x = \frac{1}{2}$ that has a price $\frac{1}{2}$ but requires a waiting time $y = \frac{1}{2}$
\end{itemize}
Under this menu, the low-risk type, finding waiting too costly, purchases the full-insurance plan; the high-risk type, finding the low-coverage plan cheap, purchases the low-coverage plan. Note that similar to \Cref{ex:negcor}, under this mechanism, the insurance allocation $x$ is \textit{not} monotone in the type $\theta^A$, which cannot happen in standard settings, but is a common feature in multidimensional settings.  

However, this mechanism is strictly dominated by simply selling the full-insurance plan to both types at a price $2$ without any costly screening. Indeed, under the proposed two-option menu, the principal gets a profit
\[\pushQED{\qed} 
\frac{1}{2} \times 2 + \frac{1}{2}  \times \big(\frac{1}{2} - \frac{1}{2} \times \frac{5}{2}\big) = \frac{5}{8} < \frac{3}{4}\,.\qedhere
\popQED\] 
\end{ex}

In \Cref{ex:poscor}, there are also many other incentive-compatible mechanisms that involve costly screening, but \Cref{thm:main} shows that any such mechanism must be dominated by simply selling the full-insurance plan to both types. The intuition behind \Cref{thm:main} can be understood in two steps. 

First, note that the costly instruments are used in different ways in the two examples. The costly action in \Cref{ex:poscor} is required for the eligibility to purchase the low-coverage option targeted at the high-risk type, which helps deter the deviation by the low-risk type to mimic the high-risk type, the \textit{\textbf{upward deviation}}. The costly action in  \Cref{ex:negcor} is required for the eligibility to purchase the full-coverage option targeted at the low-risk type, which helps deter the deviation by the high-risk type to mimic the low-risk type, the \textit{\textbf{downward deviation}}. The correlation of preferences between the two components determines which direction of the deviations the costly instrument helps deter. 

Second, it turns out that, under the surplus condition, the principal can ignore the upward deviation when screening the productive component. \Cref{ex:negcor} and \Cref{ex:poscor} have the same productive component that satisfies the surplus condition: In particular, it is socially efficient for the principal to fully insure both types. Now, to see why the principal can ignore upward deviation in this two-type case, suppose that we only impose the downward incentive constraint. Then, no type wants to mimic the high-risk type, which means that the principal would assign the efficient allocation for the high-risk type (i.e. ``no distortion at the top''). But then the high-risk type would get full coverage, which means that the low-risk type can only get weakly lower coverage. Under the downward incentive constraint, the principal would charge the prices such that the high-risk type is indifferent between getting the full-coverage plan and mimicking the low-risk type. Given that the low-risk type has a weakly lower coverage, the increasing differences property in the agent's utility then implies that the low-risk type will not deviate upward, satisfying the upward incentive constraint. 

Because the costly instrument is surplus destructive, whereas the utility through monetary transfers is perfectly transferable, the principal would only use the costly instrument to supplement the use of monetary transfers if the costly instrument can help relax some of the incentive constraints. By the first observation, if the agent has positively correlated preferences, then the costly instrument can only help relax the upward incentive constraint. But by the second observation, there is no additional benefit for the principal to relax the upward incentive constraint since it can be safely ignored under the surplus condition. Thus, the optimal mechanism does not involve costly screening.

\begin{rmk}\label{rmk:surplus}
Consistent with the above intuition, in the absence of the surplus condition, the principal cannot simply ignore the possible upward deviations for screening the productive component, and hence may make use of costly screening even if the agent has positively correlated preferences. See \Cref{ex:ubind} in \Cref{app:example} for such an example. 
\end{rmk}

\subsection{Discussion of Intuition beyond Two Types}\label{subsec:intuition_proof}

In general, the agent can have multiple and potentially multidimensional types $(\theta^A, \theta^B)$. To understand our result, there are two additional key intuitions in the more general case. 

\paragraph{Orthogonal Decomposition.}\hspace{-2mm}With multidimensional types, for any full-support distribution $\mu \in \Delta(\Theta)$, the statistical correlation between $(\theta^A, \theta^B)$ does not affect the set of incentive constraints. However, as \Cref{thm:main} and \Cref{prop:converse} show, the correlation is crucial to determine the structure of the optimal mechanism. This is because it influences the set of \textit{\textbf{binding}} incentive constraints at the optimum by influencing the principal's payoff. One way to identify which incentive constraints are binding is to consider \textit{\textbf{revealing}} part of the agent's private information to the principal instead of asking the principal to elicit this part of the information. This forms a relaxed problem since the principal must be weakly better off by observing this information. 
If it happens to be the case that, upon observing this information, the principal's optimal mechanism actually does not depend on this piece of information, then the mechanism must be optimal in the original problem and the agent does not obtain any information rent from having this piece of information. 

When $(\theta^A, \theta^B)$ are positively correlated, by the discussion in \Cref{subsec:intuition}, intuitively, the binding incentive constraints involve downward deviations in the $\theta^A$ dimension. To use the above argument, we therefore look for a piece of the agent's private information that is independent of $\theta^A$. Revealing $\theta^B$ itself would not work in general since $\theta^B$ is correlated with $\theta^A$, and hence the principal would design the mechanism depending on the information revealed. However, we can use the technique of ``orthogonal decomposition'' from the dynamic mechanism design literature (\citealt{Eso2007}; \citealt*{Pavan2014}) by writing the agent's multidimensional types as $(\theta^A, \varepsilon)$ where the transformed type $\varepsilon$, which we will reveal to the principal, is in fact independent of $\theta^A$.\footnote{The dynamic mechanism literature also studies conditions under which the principal observing $\varepsilon$ or having to elicit $\varepsilon$ leads to the same solution (see e.g. \citealt{esHo2017dynamic}). One way to interpret our construction is that we also provide a condition under which the principal observing $\varepsilon$ or having to elicit $\varepsilon$ leads to the same solution, though in a multidimensional screening problem both $\theta^{A}$ and $\varepsilon$ must be elicited simultaneously rather than sequentially.} Indeed, by a monotone coupling argument (\Cref{lem:decomp}), if $(\theta^A, \theta^B)$ are positively correlated according to our definition, then there exists $\varepsilon \perp \theta^A$ such that $\theta^B = h(\theta^A; \varepsilon)$ where $h$ is nondecreasing in $\theta^A$. Then, after $\varepsilon$ is revealed, the residual private information is one-dimensional and comonotonic just like in \Cref{ex:poscor} but potentially with a continuum of types $\theta^A$. Indeed, the binding incentive constraints would come from the downward deviations along each of these monotone paths.  

\paragraph{Downward Sufficiency.}\hspace{-2mm}When the agent has multiple types of $\theta^A$, it turns out that downward incentive constraints continue to be sufficient for all incentive constraints at an optimal solution in the one-dimensional screening problem. This downward sufficiency theorem is our main technical result and relies on the surplus condition. In the above two-type example, as we have discussed, it is relatively easy to see why the downward incentive constraint is sufficient for obtaining the optimal mechanism when screening the productive component. Another case where it is relatively easy to see that only downward incentive constraints matter is when the \textit{\textbf{virtual surplus function}} satisfies  increasing differences (pp. 262–268 of \citealt{Fudenberg1991})---in that case, the Myersonian relaxation which keeps only \textit{\textbf{local}} downward incentive constraints leads to a pointwise solution that is monotone. Note that such an assumption on the virtual surplus function is a joint condition on the utility functions and the type distribution of $\theta^A$, whereas our surplus condition is agnostic to the type distribution. Perhaps surprisingly, even when the Myersonian relaxation leads to a non-monotone solution, the set of \textit{\textbf{global}} downward incentive constraints turns out to be sufficient under the surplus condition.

The key intuition behind this result can be understood as follows. Suppose that we impose only the downward incentive constraints. If it happens to be the case that the resulting optimal allocation rule (in the productive component alone) is monotone, then the previous argument in the two-type case extends and shows that the mechanism would satisfy all incentive constraints.
It turns out that the surplus condition would lead the principal to design a monotone allocation given \textit{\textbf{global}} downward incentive constraints. Roughly speaking, the intuition can be understood as follows. When the allocation rule is not monotone, there will be some middle type $\theta^A_{m}$ whose allocation is lower than the allocation of some lower type $\theta^A_{l} < \theta^A_m$. But then it means that the higher types $\theta^A_h > \theta^A_m$ would not want to deviate to $\theta^A_m$ even if we perturb the allocation of $\theta^A_m$---importantly, here the binding downward incentive constraint is nonlocal---the most tempting deviation from $\theta^A_h$ is to $\theta^A_{l}$. This implies that when perturbing the allocation of $\theta^A_m$, there is no information rent concern but only surplus concern. But the surplus condition ensures that we prefer a monotone allocation rule when maximizing surplus, which leads to a profitable perturbation. Hence, the principal would optimally choose a monotone allocation rule given the set of all downward incentive constraints.

\begin{rmk}
In one-dimensional screening problems, even if the Myersonian approach that keeps only the local downward incentive constraints leads to a non-monotone allocation rule, one can resort to ironing because monotonicity of the allocation rule is implied by implementability. However, as shown in \Cref{ex:poscor}, with a multidimensional allocational space, even the allocation rule in the productive component need not be monotone. The downward sufficiency theorem shows that there is an underlying deeper property in one-dimensional screening problems that can be exploited to prove optimality of certain one-dimensional mechanisms in multidimensional settings.\footnote{An alternative approach explored in \citet{yang2023nested} is to provide more general conditions under which the Myersonian relaxation succeeds by developing more general monotone comparative statistics theorems.} 
\end{rmk}

\section{Proof Sketch for the Main Result} \label{sec:proof}

In this section, we sketch the proof of the first part of \Cref{thm:main}, leaving the second part (the uniqueness part) and additional details to the appendix. 

The proof follows closely the intuition provided in \Cref{subsec:intuition}. It consists of three steps. \Cref{subsec:decompose} considers a relaxed problem by decomposing the multidimensional type space into a collection of one-dimensional paths. \Cref{subsec:reconstruction} shows that on each path, for any mechanism involving costly screening, we can reconstruct a one-dimensional mechanism involving no costly screening and satisfying all downward incentive constraints. \Cref{subsec:dst} proves the downward sufficiency theorem---the optimal mechanism subject to all downward incentive constraints is fully incentive compatible---which concludes the proof of \Cref{thm:main}.

\subsection{Monotone Path Decomposition} \label{subsec:decompose}

We start by considering a relaxed problem where we reveal part of the private information of the agent to the principal. As discussed in \Cref{subsec:intuition}, we will construct a random variable $\varepsilon$ such that \textit{(i)}  $\varepsilon$ is independent of the productive type $\theta^A$, and \textit{(ii)} the residual private information of the agent $(\theta^A, \theta^B) \mid \varepsilon$, once $\varepsilon$ is revealed, is one-dimensional (and comonotonic). By a classic result of \citet{Strassen1965} on monotone coupling, stochastic monotonicity implies that we can always find such a monotone path decomposition: 

\begin{lemma}[Measurable monotone coupling] \label{lem:decomp}
If $\theta^B$ is stochastically nondecreasing in $\theta^A$, then there exist a measurable space $\mathcal{E}$; an $\mathcal{E}$-valued random variable $\varepsilon$, independent of $\theta^A$; and a measurable function $h: \Theta^A \times  \mathcal{E} \rightarrow \Theta^B$ nondecreasing in the first argument such that 
\[\theta \eqid (\theta^A, h(\theta^A; \varepsilon))\,.\]
\end{lemma}

The proof of this lemma is in \Cref{app:add}, building on a result of \citet{kamae1978stochastic}.\footnote{Our use of monotone coupling follows \citet{haghpanah2021pure}, though because of measurability issues in our setting, we prove a measurable monotone coupling lemma.} To illustrate \Cref{lem:decomp}, consider the case where $\Theta^B$ is one-dimensional; in this case, this lemma follows by using an ``orthogonal decomposition'' as follows.\footnote{See \citet{Eso2007} and  \citet*{Pavan2014}.} Let $\varepsilon$ be an independent uniform $[0, 1]$ draw. 
Then the random vector $(\theta^A, \theta^B)$ can be simulated by 
\[(\theta^A, \theta^B) \eqid (\theta^A, F^{-1}(\varepsilon \mid \theta^A)) \,,\]
where $F^{-1}(\, \cdot \, | \, \theta^A)$ is the generalized inverse function of $F(\, \cdot \, | \,\theta^A)$. Our positive correlation condition states that $\theta^B \mid \theta^A$ shifts upward in the sense of stochastic dominance as $\theta^A$ increases. This implies that  $F^{-1}(\varepsilon \,|\, \cdot \,  )$ is a nondecreasing function. Hence letting $h(\theta^A;\varepsilon) := F^{-1}(\varepsilon\mid \theta^A)$ gives a construction satisfying \Cref{lem:decomp}. However, when $\Theta^B$ is multidimensional, the existence of such a monotone coupling is nonconstructive and relies on Strassen's theorem.

By \Cref{lem:decomp}, the agent's type can be equivalently represented as $(\theta^A, \varepsilon)$. Now, consider the relaxed problem where we reveal the information $\varepsilon$ to the principal. Formally, let $\Theta_\varepsilon :=\{(\theta^A,\theta^B): \theta^B = h(\theta^A; \varepsilon), \theta^A \in \Theta^A\}$ be the decomposed monotonic path given a realization $\varepsilon$. Let $\mathcal{M}(\Theta_\varepsilon)$ be the set of IC and IR mechanisms with type space $\Theta_\varepsilon$.

For each realization $\varepsilon$, the screening problem 
\[\label{eq:sobj} \sup_{(x,y,t) \in \mathcal{M}(\Theta_\varepsilon)} \E\big [v^A(x(\theta), \theta^A) + v^B(y(\theta), \theta^B) + t(\theta) \mid \varepsilon \big ] \tag{3}\]
is precisely the screening problem where the information $\varepsilon$ is revealed to the principal. Clearly, the principal would be weakly better off by receiving the information. 

We will show that even in this relaxed problem, for each realization $\varepsilon$, the principal does not use the costly instruments. Then, because \textit{(i)} $\varepsilon$ is independent of $\theta^A$, and \textit{(ii)} the preferences are additive and 
\[v^B(y_0, \theta^B) = u^B(y_0, \theta^B) = 0\,,\]
there exists a solution to the relaxed problem such that it does not depend on $\varepsilon$, and hence is implementable in the original problem where $\varepsilon$ has to be elicited.

 Conditional on $\varepsilon$, by \Cref{lem:decomp}, the screening problem \eqref{eq:sobj} is an instance of our problem where the agent's types are comonotonic. In fact, it is then without loss of generality to assume that $\theta^A = \theta^B$ because we can define $\tilde{u}^B_{\varepsilon}(y, \theta^A):=u^B(y, h(\theta^A; \varepsilon))$ and $\tilde{v}^B_{\varepsilon}(y, \theta^A) := v^B(y, h(\theta^A; \varepsilon))$ and study the screening problem with utility functions $(u^A, \tilde{u}^B_\varepsilon, v^A, \tilde{v}^B_\varepsilon)$ (which continues to satisfy all of our assumptions since $h$ is nondecreasing). 
 
 Therefore, to prove \Cref{thm:main}, it suffices to prove it for the perfectly correlated case where $\theta^A = \theta^B$. In what follows, we assume $\Theta$ is one-dimensional and write $\theta = \theta^A = \theta^B$.

\subsection{Reconstruction Lemma} \label{subsec:reconstruction}

Now, suppose that we are given a mechanism $(x, y, t): \Theta \rightarrow \mathcal{X} \times \mathcal{Y} \times \R$, where $\Theta$ is one-dimensional (i.e., $\theta^B=\theta^A = \theta$). Consider the following reconstruction: 
\[\tilde{x}(\theta) = x(\theta)\,, \qquad \tilde{y}(\theta) = y_0\,, \qquad \tilde{t}(\theta) = t(\theta) - u^B(y(\theta), \theta) \,.\]
The reconstruction maintains the same allocations for the productive component, involves no costly screening, and uses transfers to keep all types at their previous utility levels, \textit{assuming} they report truthfully. 

Assuming truthful reporting, this increases the total surplus while giving the same surplus to the agent, and therefore increases the principal's payoff. Indeed, the change in principal's payoff is  
\[\E\big [v^B(y_0, \theta) - v^B(y(\theta), \theta) - u^B(y(\theta), \theta) \big] = \E\big [-s^B(y(\theta), \theta)\big ] \geq 0\,.\]
The last inequality is strict if $\P(y(\theta) \neq y_0) > 0$ and the instruments are strictly costly.

Because the reconstruction maintains the utility for each type under truthful reporting, $(\tilde{x}, \tilde{y}, \tilde{t})$ satisfies all IR constraints. However, this mechanism is not necessarily IC. Indeed, suppose for illustration that $u^B(y, \,\cdot\,)$ is strictly increasing, and for some $\hat{\theta} > \theta$, $\IC[\theta \rightarrow \hat{\theta}]$ binds under $(x, y, t)$. Consider the same deviation under $(\tilde{x}, \tilde{y}, \tilde{t})$:
\begin{align*}  \label{eq:upIC}
    u^A(\tilde{x}(\theta), \theta)  - \tilde{t}(\theta)  &=   u^A(x(\theta), \theta) +  u^B(y(\theta), \theta) - t(\theta) \\
    &=   u^A(x(\hat{\theta}), \theta) +  u^B(y(\hat{\theta}), \theta) - t(\hat{\theta}) \\
    &<   u^A(x(\hat{\theta}), \theta) +  u^B(y(\hat{\theta}),\hat{\theta}) - t(\hat{\theta}) \\
    &=   u^A(\tilde{x}(\hat{\theta}), \theta) -  \tilde{t}(\hat{\theta})\,, \tag{4}
\end{align*}
where the first and the last line follow by construction, the second line uses the binding IC constraint, and the third line uses that $u^B(y, \,\cdot\,)$ is strictly increasing. Therefore, $\IC[\theta \rightarrow \hat{\theta}]$ is not satisfied under $(\tilde{x}, \tilde{y}, \tilde{t})$. 

This demonstrates that the reconstruction does not work directly. However, the same reasoning also shows that all downward IC constraints are still satisfied after this reconstruction. Indeed, consider a downward deviation $[\theta \rightarrow \hat{\theta}]$ for any $\hat{\theta} < \theta$:
\begin{align*} \label{eq:downIC}
    u^A(\tilde{x}(\theta), \theta)  - \tilde{t}(\theta)  &=   u^A(x(\theta), \theta) +  u^B(y(\theta), \theta) - t(\theta) \\
    &\geq   u^A(x(\hat{\theta}), \theta) +  u^B(y(\hat{\theta}), \theta) - t(\hat{\theta}) \\
    &\geq   u^A(x(\hat{\theta}), \theta) +  u^B(y(\hat{\theta}),\hat{\theta}) - t(\hat{\theta}) \\
    &=   u^A(\tilde{x}(\hat{\theta}), \theta) -  \tilde{t}(\hat{\theta})\,, \tag{5}
\end{align*}
where the first and the last line follow by construction, the second line follows from $(x, y, t)$ being IC, and the third line follows from that $u^B(y, \,\cdot\,)$ is nondecreasing. Therefore, $(\tilde{x}, \tilde{y}, \tilde{t})$ satisfies all downward IC constraints. 

Let $\tilde{\M}(\Theta)$ denote the space of measurable maps $\Theta \rightarrow \X \times \Y \times \R$ that are \textit{(i)} IR, \textit{(ii)} involve no costly screening, and \textit{(iii)} satisfy all downward IC constraints. 

The reconstruction argument gives the following lemma:
\begin{lemma}\label{lem:dom}
Consider any $(x, y, t) \in \M(\Theta)$. There exists some $(\tilde{x}, \tilde{y}, \tilde{t}) \in \tilde{\M}(\Theta)$ such that 
\[\E\big[v^A(x(\theta), \theta) + v^B(y(\theta), \theta) + t(\theta)\big]\leq \E\big[v^A(\tilde{x}(\theta), \theta) + v^B(\tilde{y}(\theta), \theta) + \tilde{t}(\theta)\big]\,.\]
If the instruments are strictly costly and $\P(y(\theta) = y_0) < 1$, then the above inequality is strict.   
\end{lemma}

By \Cref{lem:dom}, it is, therefore, always an \textit{upper bound} for the principal to optimize over $\tilde{\M}(\Theta)$.  Because $v^B(y_0, \theta) = u^B(y_0, \theta) = 0$, the principal then solves the following:
\begin{alignat*}{2}\label{eq:1d}
\sup_{(x, t):\text{ }\Theta \rightarrow \X \times \R, \text{ measurable}} \E[&v^A(x(\theta), \theta) + t(\theta)]  \tag{6} \\
\text{subject to}\quad u^A(x(\theta), \theta) - t(\theta) &\geq u^A(x(\hat{\theta}), \theta)  - t(\hat{\theta})  \quad && \text{for all } \theta > \hat{\theta}\,, \\
u^A(x(\theta), \theta) - t(\theta) &\geq 0 && \text{for all } \theta\,. 
\end{alignat*}
This problem is a one-dimensional screening problem except that all upward IC constraints are ignored. For future reference, we use $\eqref{eq:1d}^\dagger$ to denote the version of problem \eqref{eq:1d} with both the downward and upward IC constraints. 

If we can show that there exists $(x^*, t^*)$ solving problem \eqref{eq:1d} and satisfying also all upward IC constraints, then  \Cref{thm:main} follows. From now on, we drop the superscript $A$ whenever clear, as we will focus only on the productive component. 

\subsection{Downward Sufficiency Theorem} \label{subsec:dst}

\begin{theorem}[Downward sufficiency]\label{thm:dbind}
Consider any one-dimensional screening problem. Under the surplus condition, there exists an optimal solution to \eqref{eq:1d} that also satisfies all upward IC constraints. 
\end{theorem}

We sketch the proof of \Cref{thm:dbind} for the case of finite $\Theta$. The appendix provides details and extends the argument to the general case by approximation. The proof does not directly work with a continuous type space because, as discussed in \Cref{subsec:intuition_proof}, it necessarily relies on both the local and the nonlocal downward incentive constraints. It turns out that \textit{any} allocation rule can be implemented via transfers to satisfy all downward incentive constraints, but the transfer formula differs from the one implied by the Envelope theorem whenever the allocation rule is not monotone. Thus, to prove \Cref{thm:dbind}, we characterize the transfer formula under downward incentive constraints in full generality and then show that, while non-monotone allocation rules are feasible, they are suboptimal under the surplus condition. 

Suppose $|\Theta| = n < \infty$  and order types by $\theta_1 < \theta_2 < \cdots < \theta_n$. As before, let $\mu \in \Delta(\Theta)$ denote the distribution of $\theta$. In this finite-type case, we assume that $\mu$ has full support. Without loss of generality, suppose $0 \leq \theta_1$ and $\theta_n \leq 1$. A mechanism is specified by $(x_1, x_2, \dots, x_n)$ and $(t_1, t_2, \dots, t_n)$. The principal's problem is:
\begin{alignat*}{2} \label{eq:finite}
\max_{(x, t)\in \X^n \times \R^n} \sum_i \mu(\theta_i) (&v(x_i, \theta_i) + t_i)  \tag{7} \\
\text{subject to} \quad u(x_i, \theta_i) - t_i &\geq  u(x_j, \theta_i) - t_j  \quad && \text{for all } i > j\,, \\
u(x_i, \theta_i) - t_i &\geq 0 && \text{for all } i\,. 
\end{alignat*}

We replace $\sup$ with $\max$ as the existence of the solution is easy to see by compactness arguments. In this finite-type case, we prove a stronger claim than \Cref{thm:dbind} by showing that every optimal downward IC mechanism must satisfy all upward incentive constraints. The proof consists of two steps: We first show that for any fixed $x\in \mathcal{X}^n$, there exists a unique optimal transfer $t$ such that $(x, t)$ is downward IC (\textbf{Step 1}), and then show by contradiction that any optimal solution $(x, t)$ to \eqref{eq:finite} must also satisfy the upward IC constraints using a perturbation argument (\textbf{Step 2}).

We start with some carefully chosen definitions. Fix any allocation rule $x \in \mathcal{X}^n$. Let
\[S := \{i: x_i \geq x_j \text{ for all } j < i\} \cup \{n\}\]
be the \textit{\textbf{running maximum index set}} of $x$ (including the last index). A \textbf{$U$\textit{-shaped region}} $r$ is a set of indices of the form  $\{i, i+1, \dots, i'\}$ such that \textit{(i)} $i$ and $i'$ are two consecutive elements of $S$, and \textit{(ii)} $x_i > x_{i+1}$. By definition, there is a finite sequence $\mathcal{L}$ of $U$-shaped regions. Let $L:= |\mathcal{L}|$ be the number of $U$-shaped regions. We write $o_l$ and $d_l$ as the starting and end index of the $l$-th $U$-shaped region $r_l$. An \textit{\textbf{optimal downward transfer operator}} (or simply \textit{\textbf{transfer operator}}) is a map $\mathbf{T}: \mathcal{X}^n \rightarrow \R^n$ such that for any $x \in \mathcal{X}^n$, $\mathbf{T}[x]$ is an optimal solution to the following problem: 
\begin{alignat*}{2} \label{eq:tproblem}
\max_{ t \in \R^n} \sum_i \mu(\theta_i) (&v(x_i, \theta_i) + t_i)\,  \tag{8} 
\end{alignat*}
subject to the same downward IC and IR constraints as in \eqref{eq:finite}. That is, $\mathbf{T}[x]$ maps an allocation rule $x$ to the optimal transfer rule that implements $x$ in a downward IC fashion. At this stage, it is unclear whether such an operator $\mathbf{T}$ is defined for all $x \in \mathcal{X}^n$. 

\paragraph{Step 1.}\hspace{-2mm}We first show that such an operator exists and is unique. In particular, we explicitly characterize $\mathbf{T}$, which will be used in the perturbation argument in \textbf{Step 2}. For notational convenience, we write $\IC[i \rightarrow j]$ (or simply $[i \rightarrow j]$) and $\IR[i]$ as shorthand for $\IC[\theta_i \rightarrow \theta_j]$ and $\IR[\theta_i]$, respectively. 
\begin{claim}\label{claim:transfer} There exists a unique transfer operator $\mathbf{T}$. For any $x \in \mathcal{X}^n$, $\mathbf{T}[x]$ is the solution to the system of equations defined by the following constraints with equality: $\textnormal{IR}[1]$, and 
\[\textnormal{IC}\big[i \rightarrow \max\{j \in S: j < i\}\big]\,,\]
for all $i$. In particular, for all $i$, 
\[(\mathbf{T}[x])_i =   u(x_i, \theta_i) - \sum_{j \in S \cup \{i\}:\, j \leq i}\Big(u(x_{\max\{k \in S: k < j\}}, \theta_{j}) - u(x_{\max\{k \in S: k < j\}}, \theta_{\max\{k \in S: k < j\}})\Big)\,. \label{eq:t} \tag{9}\]
\end{claim}

In words, for an arbitrary $x \in \X^n$, there exists a unique optimal $t$ subject to that $(x, t)$ is IR and downward IC. Moreover, for a given $x$, the binding IC constraints for $t$ go from every index $i > 1$ to the largest element in $S$ that is strictly less than $i$. \Cref{fig:transfer} illustrates how the $U$-shaped regions and the binding constraints in \Cref{claim:transfer} are identified. The proof of \Cref{claim:transfer} is in the appendix. The key intuition behind \Cref{claim:transfer} is that for a fixed allocation rule, because of the increasing differences property of the agent's preferences, the binding IC constraints are such that each type $\theta_i$ must be indifferent from taking their own allocation and taking the \textit{maximum} allocation assigned to some type below them. 

\begin{figure}
    \centering
\hspace{-0.07\linewidth}
\begin{subfigure}[b]{0.48\linewidth} 
\centering
        \includegraphics[scale=0.9]{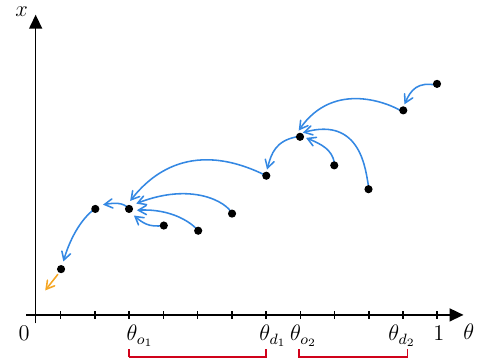}    \caption{Binding constraints for fixed allocations}
            \label{fig:transfer}
\end{subfigure}
\hspace{0.00\linewidth}
\begin{subfigure}[b]{0.48\linewidth}
    \centering
        \includegraphics[scale=0.9]{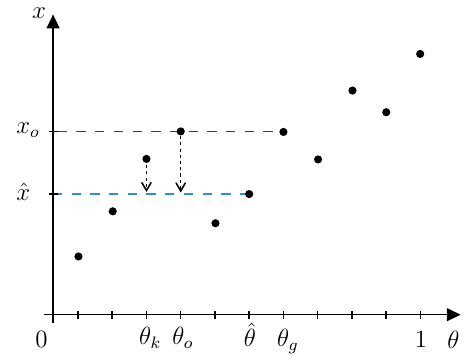}
 \caption{Perturbation of the allocations}
            \label{fig:allocation}
\end{subfigure}
       \caption{Illustration of the proof of \Cref{thm:dbind}}
\end{figure}

As the figure shows, equivalently, the local IC constraints bind until one gets into a $U$-shaped region (beginning with say index $o$) where the binding constraints then point toward $\theta_{o}$. Therefore, the formula \eqref{eq:t} can also be written as 
\[\label{eq:t2} (\mathbf{T}[x])_i = u(x_i, \theta_i) - \underbrace{\sum_{j=1,2,\dots, i-1: \ j \in Q} \Big( u(x_{j}, \theta_{j+1}) - u(x_{j}, \theta_{j}) \Big) }_\text{local} - \underbrace{\sum_{l=1,2,\dots, L: \  o_l < i} \Big( u(x_{o_l}, \theta_{\min(d_l, i)}) - u(x_{o_l}, \theta_{o_l}) \Big) }_\text{nonlocal} \,. \]
where $Q$ is the region where $x$ is monotone in an appropriate sense.\footnote{Formally, $Q = \big \{1\leq j \leq n: j \not \in r_l \text{ for all $l$, or } j = d_l \neq o_{l+1} \text{for some $l$} \big\}$.} The first sum arises from the binding local downward IC constraints, and the second sum arises from the binding nonlocal downward IC constraints. If $x$ is monotone, then \eqref{eq:t} will reduce to the standard transfer formula in one-dimensional screening models, in which only the \textit{local} downward IC constraints matter for transfers.\footnote{In particular, if there is a continuum of types, then it would coincide with the one given by the Envelope theorem.} However, importantly, we cannot rule out non-monotone allocation rules $x$ by implementability here, since we do not have upward incentive constraints.

\paragraph{Step 2.}\hspace{-2mm}Suppose for contradiction that there exists an optimal $(x, t)$ to \eqref{eq:finite} that does not satisfy some upward IC constraints. We construct a perturbation that strictly improves on $(x, t)$. There are two cases: \textit{(i)} $x$ is monotone and \textit{(ii)} $x$ is not monotone. The first case is simple. As noted before, when $x$ is monotone, $\mathbf{T}[x]$ reduces to the standard transfer formula and hence $(x, \mathbf{T}[x])$ must satisfy all incentive constraints including the upward ones. Then, $t$ cannot be identically equal to $\mathbf{T}[x]$. But that implies $(x, t)$ can be strictly improved by $(x, \mathbf{T}[x])$ by \Cref{claim:transfer}. 

Now, consider the second case where $x$ is not monotone. Then, there exists a $U$-shaped region. \Cref{fig:allocation} illustrates our perturbation. Specifically, it is constructed as follows. 

Let $r$ be the first $U$-shaped region, $o$ its starting index, and $d$ its end index. Let
\[g := \min \big\{j > o: x_j \geq x_o \big\}\]
denote the first index after $o$ with associated allocation no less than $x_o$. Put $g = n+1$ if the above set is empty. Then, either $g = d$ or $g = n+1$. Let
\[\hat{x} := \max \big\{x_j: o < j < g \big\}\]
denote the largest allocation for indices strictly between $o$ and $g$. We have $\hat{x}<x_o$. Let $j^* \in \{o+1, \dots, g-1\}$ be the first index achieving the above maximum and let $\hat{\theta} = \theta_{j^*}$. Let 
\[k := \min \big \{j:  x_j > \hat{x} \big \} \]
denote the first index whose associated allocation is strictly higher than $\hat{x}$. Since $x_o > \hat{x}$, we have $k \leq o$. Because $r$ is the first $U$-shaped region, we have $\hat{x}  <  x_{k} \leq x_{k+1} \leq \cdots \leq  x_o$. Consider the following perturbation $\tilde{x}$: for all $j \in \{1, \dots, n\}$, 
\[\tilde{x}_j := \begin{cases}
\hat{x} & \text{ if } j \in \{ k, k+1, \dots, o\}\,; \\
x_j & \text{ otherwise.} \\
\end{cases}\]

\begin{claim} \label{claim:perturb}
Under the surplus condition, $(\tilde{x}, \mathbf{T}[\tilde{x}])$ strictly increases the objective value of \eqref{eq:finite}. 
\end{claim}

The proof is in the appendix. The key intuition can be understood as follows. Suppose $k = o$ for simplicity. We explain why perturbing the allocation of $\theta_o$ downward (as in \Cref{fig:allocation}) is profitable. For the types within the interior of the $U$-shaped region, there is no higher type who wants to deviate to mimicking these types even if we perturb the allocations of these types (recall \Cref{fig:transfer}). This implies that, at least locally, when we perturb the allocations of these types, there is no information rent concern but only surplus concern. But then since we start from the optimal solution, for type $\hat{\theta}$ who is in the $U$-shaped region, the assigned allocation $\hat{x}$ must have generated weakly more surplus than the alternative $x_o$ which could have been assigned to $\hat{\theta}$. Since $\hat{x} < x_o$ and $\theta_o < \hat{\theta}$, the surplus condition implies that the same surplus comparison also holds for type $\theta_o$. But then it is profitable to perturb the allocation of $\theta_o$ downward from $x_o$ to $\hat{x}$, since that increases the surplus \textit{and} reduces the information rents collected by the types above $\theta_o$. 

By \Cref{claim:perturb}, we can strictly improve on $(x, t)$. But $(x, t)$ is assumed to be optimal, contradiction. So any optimal solution to \eqref{eq:finite} satisfies all upward IC constraints, proving \Cref{thm:dbind}. By \Cref{subsec:decompose,subsec:reconstruction}, \Cref{thm:main} then follows.

\section{Applications}\label{sec:app}

This section presents three applications. \Cref{subsec:csignal} shows an application to monopoly pricing with costly signals. \Cref{subsec:labor} shows an application to labor market screening. \Cref{subsec:bundling} shows how to apply our result to multiproduct pricing and optimal bundling. 

\subsection{Monopoly Pricing with Costly Signals} \label{subsec:csignal}

A monopolist sells a quality-differentiated spectrum of goods. A buyer of type $\theta^A \in \Theta^A$ receives utility $u^A(x, \theta^A)$ from consuming the good of quality $x \in \X$. The seller incurs a cost $C(x, \theta^A)$ to produce the good of quality $x$ for type $\theta^A$. Suppose $u^A(x, \theta^A)$ is nondecreasing in $\theta^A$ and has strict increasing differences, and that the surplus function $u^A(x, \theta^A) - C(x, \theta^A)$ satisfies our surplus condition. 

Besides offering a menu of products of different qualities and prices, the monopolist can make the offers contingent on various costly signals (e.g. waiting in line, collecting coupons, walking up stairs). A costly signal is represented by $y \in [0, 1]$. To obtain a signal $y$, a buyer of type $\theta^B \in \Theta^B$ incurs a cost $c(y, \theta^B)$ that is nonincreasing in $\theta^B$ with $c(0, \theta^B) = 0$ (so $\theta^B$ represents the willingness to endure various costly activities). 

\Cref{thm:main} then says that if $\theta^B$ is positively correlated with $\theta^A$ according to our notion, then the monopolist never makes more profits by using these costly signals.  
\begin{prop}\label{prop:pricing}
If $\theta^B$ is stochastically nondecreasing in $\theta^A$, then there exists a profit-maximizing mechanism that requires no costly signals. 
\end{prop}
As a consequence of \Cref{prop:pricing}, if we observe that a seller in practice uses these costly instruments, then we should expect that the consumers with higher willingness to pay for the seller's goods tend to incur higher costs to obtain the signals.

\subsection{Labor Market Screening} \label{subsec:labor}

A monopsonistic firm wants to hire a worker to perform a task. The firm gets a payoff $V(\theta^A)$ for hiring a worker of ability $\theta^A \in \Theta^A \subset \R$, where $V$ is continuous, nondecreasing in $\theta^A$. The worker suffers a cost $C(\theta^A)$ for performing the task, where $C$ is continuous, strictly decreasing in $\theta^A$. Let $x \in [0, 1]$ be the probability of hiring the worker. 

The firm can ask the applicant to obtain a credential. The applicant gets a negative payoff $\theta^B \cdot y \leq 0$ for obtaining a credential of level $y\in[0,1]$. Both $\theta^A$ and $\theta^B$ are the worker's private information. For a given wage level $w$, hiring probability $x$, and credential level $y$, the firm's payoff is $(V(\theta^A) - w)x$, and the worker's payoff is $(w - C(\theta^A))x + \theta^B y$. 

\begin{prop} \label{prop:labor}
If $\theta^B$ is stochastically nondecreasing in $\theta^A$, then there exists an optimal mechanism that does not require any credential. 
\end{prop}

This result is an immediate consequence of \Cref{thm:main}. In particular, the worker has positively correlated preferences between work and credential because a higher ability type finds work easier and also tends to find obtaining the credential easier. In contrast to the classic case of competitive wages \`{a} la \citet{Spence1973Job}, who shows that costly credentials can serve as effective screening devices in a model with multiple firms, our result shows that costly screening may not be needed in a monopsonistic market. 

The difference in results can be understood as follows. In \citet{Spence1973Job}, wages are competitive and pinned down by expected output, which leaves very little room to screen types via monetary payments, and hence all the screening has to occur via costly efforts. By contrast, in our model, wages are set by the monopsonistic firm and different types have the same outside options. In particular, the firm can potentially screen the types by setting a wage such that a low-ability type finds it unprofitable to accomplish the task.\footnote{To further illustrate that the difference is driven by whether the wages are set competitively, \Cref{app:comp} considers a version of our model with competitive wages and multiple screening instruments. It shows that costly screening can appear in equilibrium. Intuitively, competition between firms generates higher outside options for higher ability types, which makes the binding incentive constraints become the upward ones instead of the downward ones, leading to the costly instruments becoming useful under positive correlation of preferences. }

\subsection{Optimal Bundling} \label{subsec:bundling}

Beyond the direct implications discussed in \Cref{subsec:csignal} and \Cref{subsec:labor}, our result turns out to also yield new insights into the classic optimal bundling problem that may at first glance appear to be unrelated. In particular, we first show how to apply our result to obtain new results on nested bundling (\Cref{subsubsec:nested}), and then show how we can further recover a recent result in the special case of pure bundling (\Cref{subsubsec:pure}).

\subsubsection{Nested Bundling} \label{subsubsec:nested}

There are $G$ different items indexed from $1$ to $G$. A \textit{\textbf{bundle}} $b \in \mathcal{B} := 2^{\{1,\dots, G\}}$ is a set of items. A \textit{\textbf{menu}} $B \subseteq  \mathcal{B}$ is a set of bundles (which we assume includes $\emptyset$).\footnote{To simplify notation, we omit the inclusion of $\emptyset$ in a menu whenever it is clear from the context.} A menu $B$ is \textit{\textbf{nested}} if the bundles in $B$ can be totally ordered by set inclusion.

For this application, we assume one-dimensional types and deterministic mechanisms (both are relaxed in the pure bundling application). A consumer has private information about their type $\theta \in \Theta \subset \R$. The value of bundle $b$ for type $\theta$ is denoted by $v(b, \theta)$. 

We assume that $v(b, \theta)$ is nondecreasing in $b$ (in set-inclusion order) and continuous, strictly increasing in $\theta$ (with $v(\emptyset, \theta) = 0$). The seller wants to maximize revenue (for simplicity, we assume zero marginal costs). A deterministic allocation rule is denoted by $\theta \mapsto a(\theta) \in \{0, 1\}^{\mathcal{B}}$, where $\sum_b a^b = 1$; a transfer rule is denoted by $\theta \mapsto t(\theta) \in \R$.  A menu $B$ is \textit{\textbf{optimal among deterministic mechanisms}} if there exists an optimal IC and IR deterministic mechanism $(a, t)$ such that $a(\theta) \in B$ for all $\theta$. 

\begin{prop} \label{prop:generalnested}
Consider any nested menu $B$. Suppose that
\begin{itemize}
    \item[(i)] \hspace{-1mm}for any $b, b' \in B$ such that $b \subset b'$, $v(b', \theta) - v(b, \theta)$ is strictly increasing in $\theta$;
    \item[(ii)] \hspace{-1mm}\mbox{for~any~$b \not \in B$,~there~exists~$b' \in B$~such~that~$b \subset b'$~and~$v(b', \theta) - v(b, \theta)$~is~nonincreasing~in~$\theta$.}
\end{itemize}
Then $B$ is optimal among deterministic mechanisms. 
\end{prop}

\Cref{prop:generalnested} is a natural consequence of \Cref{thm:main} once one views selling the nested menu $B$ as the productive component, and selling any bundle $b \not\in B$ \textit{instead of} the larger bundle $b' \supset b$ in the menu $B$ as a costly instrument for screening the consumer's type. 

\begin{proof}[Proof of \Cref{prop:generalnested}]
Order the bundles in menu $B$ by $b_1 \subset b_2 \subset \cdots \subset b_m$. For $b \in B$, let $i(b)$ be its associated index. We map this bundling problem to our main model as follows. 

Let $\mathcal{X} := \{1, 2, \dots, m\}$. Let the agent's utility on the productive component be 
\[u^A(x, \theta) := v(b_x, \theta)\,.\]
Let 
\[\mathcal{Y} := \Big \{y \in  \{0, 1\}^{\mathcal{B}\backslash B}: \sum_{b \not \in B} y^b \leq 1 \Big\}\,.\]
By condition \textit{(ii)}, there exists some function $\beta: \mathcal{B}\backslash B\rightarrow B$ such that \textit{(i)} $b \subset \beta(b)$ and \textit{(ii)} $v(\beta(b), \theta) - v(b, \theta)$ is nonincreasing in $\theta$. Let the agent's utility on the costly component be 
\[u^B(y, \theta) := \sum_{b \not \in B} \big (v(b, \theta) - v(\beta(b), \theta) \big ) y^b \leq 0\,.\]
Let the principal's utility functions for both components be identically $0$. 

Now, for any original deterministic allocation $a$ that assigns $a^b = 1$ for $b \in B$, we can replicate it by letting $x = i(b)$ and $y = 0$. For any deterministic allocation $a$ that assigns $a^b = 1$ for $b \not\in B$, we can replicate it by letting $x = i(\beta(b))$ and $y^{b} = 1$ (and $y^{b'} = 0$ for $b' \neq b$). Thus, the principal's optimal value in this transformed screening problem must be weakly higher than in the bundling problem.

Because $v(b, \theta) - v(\beta(b), \theta)$ is nondecreasing in $\theta$ for all $b \not \in B$, we have $u^B(y, \theta)$ is nondecreasing in $\theta$ for all $y \in \mathcal{Y}$. By construction and condition \textit{(i)}, $u^A(x, \theta)$ is nondecreasing in $\theta$ and has strict increasing differences. Then, by \Cref{thm:main}, there exists an optimal mechanism with $y^*(\theta) = 0$ for all $\theta$ in this screening problem. But that is implementable in the original problem, corresponding to a nested bundling mechanism with menu $B$.
\end{proof}

To illustrate \Cref{prop:generalnested}, consider two items. Let $g(\theta)$, $h(\theta)$, and $v(\theta)$ be the value of item $1$, the value of item $2$, and the value of the bundle $\{1, 2\}$ for type $\theta$, respectively.

\begin{cor} \label{cor:nested}
If $v(\theta) - g(\theta)$ is strictly increasing, and $v(\theta) - h(\theta)$ is nonincreasing, then the nested menu $\{\{1\}, \{1, 2\}\}$ is optimal among deterministic mechanisms.
\end{cor}

For example, suppose that the seller is a cable TV service provider who can offer basic channels and sports channels. \Cref{cor:nested} says that if a higher income consumer has a higher incremental value for sports channels, and a lower incremental value for basic channels, then it suffices for the seller to consider a two-tier menu that features a basic package and a premium package that includes both basic and sports channels. This result holds for any distribution of $\theta$. Despite its seeming simplicity, it is a new condition for the optimality of nested bundling (see \Cref{app:nestedexample} for a parametric example).\footnote{For example, \citet{bergemann2022optimality} obtain conditions for nested bundling under additive values, whereas this application allows for non-additive values; \citet{yang2023nested} also obtains conditions for nested bundling with non-additive values but the conditions there depend on the type distribution.}

\subsubsection{Pure Bundling} \label{subsubsec:pure}

In the special case of pure bundling, where the menu $B$ consists of only the \textit{\textbf{grand bundle}} $\overline{b}:=\{1,\dots,G\}$, \Cref{prop:generalnested} says that: 
\begin{cor} \label{cor:pure}
If $v(\overline{b}, \theta)-v(b,\theta)$ is nonincreasing in $\theta$ for all $b \neq \emptyset$, then pure bundling is optimal among deterministic mechanisms. 
\end{cor}

However, in this special case, we can obtain a stronger result, in fact recovering a recent result of \citet{haghpanah2021pure} on the optimality of pure bundling. 

We generalize the bundling problem in \Cref{subsubsec:nested} by allowing for \textit{(i)} multidimensional types and \textit{(ii)} stochastic mechanisms. A consumer has private information about their type $v := (v^b)_{b \in \mathcal{B}}$ where $v^b$ is their value for bundle $b$. We assume $v^{b}\leq v^{\overline{b}}$ and $v^\emptyset = 0$. A stochastic allocation rule is denoted by $v \mapsto \alpha(v) \in \Delta(\mathcal{B})$; a transfer rule is denoted by $v \mapsto t(v) \in \R$.  We continue to assume zero costs. We say that \textit{\textbf{pure bundling is optimal}} if posting a single price for the grand bundle is optimal among all stochastic mechanisms.

\begin{prop}[\citealt{haghpanah2021pure}]\label{prop:pure}
 If $\big(\frac{v^b}{v^{\overline{b}}}\big)_{b \subset \overline{b}}$ is stochastically nondecreasing in $v^{\overline{b}}$, then pure bundling is optimal.
\end{prop}

We show how to obtain \Cref{prop:pure} by mapping the bundling problem to our main model (with a mapping slightly different from \Cref{subsubsec:nested}). Let $\theta^A = v^{\overline{b}}$ be the value of the grand bundle. For any proper bundle $b$, let $\theta^b = v^b - v^{\overline{b}}$ be the difference in values for bundle $b$ and the grand bundle $\overline{b}$.  In words, $\theta^b$ is the negative value for getting bundle $b$ instead of $\overline{b}$. Let $N = 2^G - 1$, and let $\theta^B = (\theta^1, \dots, \theta^{N})$ be the profile of the differences. 

We use $x: \Theta \rightarrow [0,1]$ to denote the \textit{initial} probability allocation of the grand bundle, and $y: \Theta \rightarrow [0,1]^{N}$ to denote the costly instruments as follows. An assignment $y^b \in [0, 1]$ represents assigning bundle $b$ with probability $y^b$ while \textit{decreasing} the probability of the grand bundle $\overline{b}$ also by $y^b$. The consumer's payoff can be rewritten as 
\[\theta^A x +  \theta^B \cdot y - t \,.\]
For any stochastic allocation $\alpha$, we can replicate it by setting 
\[x = \sum_{b} \alpha^b\,,  \qquad y^b = \alpha^b \text{ for all } b \neq \overline{b} \,.\]
Thus, the principal's optimal value in this transformed screening problem must be weakly higher than in the bundling problem. Moreover, since
\[\frac{v^b}{v^{\overline{b}}} = \frac{\theta^A + \theta^b}{\theta^A} = 1 + \frac{\theta^b}{\theta^A}\,,\]
we know that if $\big(\frac{v^b}{v^{\overline{b}}}\big)_{b \subset \overline{b}}$ is stochastically nondecreasing in $v^{\overline{b}}$, then $r^B := \frac{1}{\theta^A} \theta^B$ is stochastically nondecreasing in $\theta^A$. Note that $r^B$ is the vector of \textit{\textbf{marginal rates of substitution (MRS)}} between the productive and the costly components. This is not exactly our positive correlation condition, but in this special case of multiplicative preferences, the positive correlation between MRS and productive types suffices for the conclusion of \Cref{thm:main}, i.e., that the optimal mechanism involves no costly screening.\footnote{A reader may also wonder how this condition relates to negative correlation of item values, since \citet{adams1976commodity} show that with additive and perfectly negatively correlated values, pure bundling is optimal. With additive values, it can be seen as a generalization of \citet{adams1976commodity}: e.g., suppose $v^1 = \xi \cdot \theta$ and $v^2 = (1-\xi)\cdot \theta$ where $\xi, \theta \in [0, 1]$ are independent. The stochastic ratio monotonicity condition holds as $v^1/v^2$ is independent of $v^1+v^2$. In the special case of known $\theta$, this is \citet{adams1976commodity}. 
} But then by linearity, we know that the optimal mechanism is a posted price for $x = 1$ (\citealt{Myerson1981}; \citealt{riley1983optimal}), which corresponds to a pure bundling mechanism in the original bundling problem. The appendix provides details.
\begin{rmk}\label{rmk:generalratio}
Using the same logic, we can also generalize the result of \citet{haghpanah2021pure} to a multiple-good monopoly problem allowing for both
probabilistic bundling and quality discrimination (\`{a} la \citealt{mussa1978monopoly}). In particular, for each bundle $b$, a consumer has value $v^b$ for getting the highest quality version of the bundle with probability one. In addition to the stochastic bundling allocation rule $\alpha$ as before, the monopolist can also adjust the quality of each bundle, captured by a quality allocation rule $v \mapsto q(v) \in [0, 1]^{\mathcal{B}}$. A type-$v$ consumer's payoff is given by \[\sum_{b} \alpha^b q^b v^b  - t\,. \]
The monopolist incurs a cost to improve the quality of a bundle, with a payoff given by
\[-\sum_{b} \alpha^b C(q^b) + t\,,\]
where $C(\,\cdot\,)$ is a continuous, nondecreasing, and convex function on $[0,1]$ with $C(0)=0$. By the same logic described above, if $\big(\frac{v^b}{v^{\overline{b}}}\big)_{b \subset \overline{b}}$ is stochastically nondecreasing in $v^{\overline{b}}$, then an optimal mechanism exists and can be implemented by a menu of prices for different qualities of the grand bundle. See \Cref{app:quality} for details. The optimal mechanism here involves price discrimination in general but does so only along the vertical dimension by offering different qualities of the grand bundle.
\end{rmk}

\section{Conclusion} \label{sec:conclusion}

This paper studies the effectiveness of costly instruments in a general multidimensional screening model. The model consists of two components: a one-dimensional productive component and a multidimensional costly component. Our main result says that if the agent's preferences are positively correlated between the two components in a suitably defined sense, then the costly instruments are ineffective---the optimal mechanism simply screens the one-dimensional productive component. 

Our proof provides clear insights into why this result holds. For a given multidimensional mechanism, we first decompose the multidimensional type space into a collection of one-dimensional paths, and then show that on each path, we can reconstruct a one-dimensional mechanism, involving no costly screening, that satisfies all downward incentive constraints and improves on the original mechanism. Finally, we show that only downward incentive constraints are needed in any one-dimensional screening model that satisfies the surplus condition. The last step, what we call the downward sufficiency theorem, also uncovers a novel property of one-dimensional screening models.

Armed with this understanding, we also show how additional results follow naturally. With negatively correlated preferences, we show a partial converse. Using the perspective of screening with costly instruments, as applications, we obtain new insights into monopoly pricing, bundling, and labor market screening.

\newpage 
\appendix
\crefalias{section}{appendix}
\section{Omitted Proofs}\label{app:proof}

The proofs follow the order in which they are referenced in the main text, except that we first prove \Cref{thm:dbind} (downward sufficiency theorem), as it will be used in the proof of \Cref{thm:main}.

\subsection{Proof of \texorpdfstring{\Cref{thm:dbind}}{}}

We proceed as follows: \Cref{subsubsec:transfer} proves \Cref{claim:transfer}; \Cref{subsubsec:perturb} proves \Cref{claim:perturb}; \Cref{subsubsec:completion} completes the proof of \Cref{thm:dbind} by extending the finite-case result to the general case via approximation. 

\subsubsection{Proof of \texorpdfstring{\Cref{claim:transfer}}{}}\label{subsubsec:transfer}

Relax all the constraints in \eqref{eq:tproblem} except the ones indicated in \Cref{claim:transfer}. We will show the following: First, these constraints must bind in the relaxed problem. Second, these constraints binding imply all downward IC constraints and all IR constraints. Third, there is a unique solution to the system of equations defined by these binding constraints. \Cref{claim:transfer} then follows. 

Note that for every $i  > 1$, there is precisely one corresponding constraint $[i \rightarrow j]$ for some $j$. If this constraint does not bind at some mechanism $(x, t)$, then simply set $\tilde{t}_i = t_i + \epsilon$ for some $\epsilon > 0$ small enough so that $[i \rightarrow j]$ still holds. This clearly increases the objective. It also does not distort other IC constraints. Indeed, the only other IC constraints this change affects are of the form $[k \rightarrow i]$ for some $k$, but  
\[u(x_k, \theta_k) - t_k \geq u(x_i, \theta_k) - t_i \geq u(x_i, \theta_k) - \tilde{t}_i \,.\]
Therefore, all the IC constraints identified in \Cref{claim:transfer} must bind. Similarly, $\IR[1]$ binds. 

Given that these constraints bind, we now show that they imply all the downward IC constraints in \eqref{eq:tproblem}. We first collect two lemmas:
\begin{lemma}[Local to global]\label{lem:log} Let $i > j > k$. If $[i \rightarrow j]$, $[j \rightarrow k]$ hold and $x_j \geq x_k$, then $[i \rightarrow k]$. 
\end{lemma}

\begin{lemma}[Global to local]\label{lem:gol} Let $i > j > k$. If $[i \rightarrow k]$, $[j \rightarrow k]$ bind and $x_j \leq x_k$, then $[i \rightarrow j]$. 
\end{lemma}
\Cref{lem:log} is standard; it follows from a revealed-preference argument using the single-crossing property of $u$. \Cref{lem:gol} appears to be new; it requires two binding IC constraints and follows from a revealed-preference argument that subtracts the two constraints. The proofs of these lemmas are deferred to the end.  

We show all downward IC constraints are satisfied by induction on the number of $U$-shaped regions $L$. When $L = 0$, all downward IC constraints hold by successively applying \Cref{lem:log} and building up from the adjacent local downward constraints. Suppose the claim holds for $L - 1$. Let us denote the last region as $r$ with starting index $o$ and end index $d$. By the inductive hypothesis, all downward IC constraints $[i \rightarrow j]$ are satisfied if $j < i \leq o$. We divide the remaining pairs $(j, i)$ with $j < i$ into two cases: 
\begin{itemize}
    \item[] Case (1): $o \leq j < i$. We make the following observations:
    \begin{itemize}
        \item[(a)] if $o \leq j < i \leq d$, then $[i \rightarrow j]$ follows by the binding IC constraints $[i \rightarrow o]$, $[j \rightarrow o]$, $x_j \leq x_o$, and \Cref{lem:gol};
        \item[(b)] if $d \leq j < i$, then $[i \rightarrow j]$ follows by successively applying \Cref{lem:log};
        \item[(c)] if $j < d < i$, then $[i \rightarrow j]$ follows by $[i \rightarrow d]$ from (b), $[d \rightarrow j]$ from (a), $x_d \geq x_j$, and \Cref{lem:log}. 
    \end{itemize}
    \item[] Case (2): $j < o < i$. Note that $x_j \leq x_o$ for all $j < o$. Then, $[i \rightarrow j]$ follows by $[i \rightarrow o]$ from Case (1), $[o \rightarrow j]$ from the inductive hypothesis, $x_o \geq x_j$, and \Cref{lem:log}.
\end{itemize}

Together, these cover all the downward IC constraints and prove the inductive step. The IR constraints follow easily from $\IR[1]$ and $\IC[i \rightarrow 1]$ and that $u(x, \,\cdot\,)$ is nondecreasing. 

The binding constraints define a system of $n$ equations for $t$. With some calculations, it is not hard to see that these equations can be solved successively starting from the lowest one. In particular, by induction, the solution is uniquely defined by \eqref{eq:t}. 

\vspace{0.5cm}

\begin{proof}[Proof of \Cref{lem:log}]
Write out $[i \rightarrow j]$ and $[j \rightarrow k]$: 
\[ u(x_i, \theta_i) - t_i \geq u(x_j, \theta_i) - t_j \,;\]
\[ u(x_j, \theta_j) - t_j \geq u( x_k, \theta_j) - t_k \,.\]
Adding these two yields 
\[ u(x_i, \theta_i) - t_i +  u(x_j, \theta_j) - t_j\geq u( x_j, \theta_i) - t_j + u( x_k, \theta_j) - t_k\,.\]
Hence,
\[ u(x_i, \theta_i) - t_i  \geq (u(x_j, \theta_i) + u(x_k, \theta_j) -  u(x_j, \theta_j)) - t_k \,.\]
Using $x_j \geq x_k$, $\theta_i > \theta_j$, and the strict increasing differences property of $u$, we have 
\[u(x_j, \theta_i) + u(x_k, \theta_j) -  u(x_j, \theta_j) \geq u(x_k, \theta_i)\,.\]
Thus $[i \rightarrow k]$ follows. 
\end{proof}
\vspace{0.5cm}

\begin{proof}[Proof of \Cref{lem:gol}]
Write out the binding constraints $[i \rightarrow k]$ and $[j \rightarrow k]$: 
\[ u(x_i, \theta_i) - t_i = u(x_k, \theta_i) - t_k \,;\]
\[ u(x_j, \theta_j) - t_j = u(x_k, \theta_j) - t_k \,.\]
Subtracting these two yields 
\[ u(x_i, \theta_i) -  u(x_j, \theta_j) - t_i  = u(x_k, \theta_i) -u(x_k, \theta_j) - t_j \,.\]
Hence,
\[u(x_i, \theta_i)  - t_i  =  (u(x_j, \theta_j) + u(x_k, \theta_i) -u(x_k, \theta_j)) - t_j \,.\]
Using $x_k \geq x_j$, $\theta_i > \theta_j$, and the strict increasing differences property of $u$, we have 
\[u(x_j, \theta_j) + u(x_k, \theta_i) -u(x_k, \theta_j) \geq u(x_j, \theta_i) \,.\]
Thus $[i \rightarrow j]$ follows.
\end{proof}

\subsubsection{Proof of \texorpdfstring{\Cref{claim:perturb}}{}}\label{subsubsec:perturb}
The main difficulty here is that $\mathbf{T}$ is a complicated operator because of the binding nonlocal incentive constraints. The perturbed allocation $\tilde{x}$ is constructed in such a way that it maintains the \textit{form} of $\mathbf{T}$ in the following sense. Define $\varphi: \mathcal{P}(\{1,\dots, n\}) \times \mathcal{X}^n \rightarrow \R^n$ by 
\[(\varphi(K, a))_i := u(a_i, \theta_i) - \sum_{j \in K \cup \{i\}:\, j \leq i}\Big(u(a_{\max\{k \in K: k < j\}}, \theta_j) - u(a_{\max\{k \in K: k < j\}}, \theta_{\max\{k \in K: k < j\}})\Big)\]
for all $i = 1,\dots, n$, where $\mathcal{P}(\,\cdot\,)$ is the power set operator. For any allocation rule $a \in \mathcal{X}^n$, let $S_a$ be the running maximum index set as defined earlier. Then, by \eqref{eq:t}, $\mathbf{T}[a] = \varphi(S_a, a)$. 

\begin{lemma}\label{lem:form}
Consider any two allocation rules $a, b \in \mathcal{X}^n$ with running maximum index sets $S_a, S_b$. Suppose $S_a \subseteq S_b$, and $b_i = b_{\max\{j \in S_a: j < \min(S_b \backslash S_a)\}}$ for all $i \in S_b \backslash S_a$.  Then,
\[\mathbf{T}[b] = \varphi(S_b, b) = \varphi(S_a, b)\,.\]
\end{lemma}

The proof of this lemma is deferred to the end. Note that by construction $x$, $\tilde{x}$ always satisfy the conditions in \Cref{lem:form}. Applying \Cref{lem:form} to $x$, $\tilde{x}$ gives $\mathbf{T}[\tilde{x}] = \varphi(S, \tilde{x})$, where $S$ is the running maximum index set of $x$. We show that the objective of \eqref{eq:finite}, after plugging in $\varphi(S, \,\cdot\,)$, weakly increases on the parts involving $x_k, x_{k+1}, \dots, x_{o-1}$ (which may be an empty set) and strictly increases on the parts involving $x_o$ (which always exist). 

Fix any $j \in \{k,k+1,\dots,o-1\}$. Plugging $\varphi(S, \,\cdot\,)$ into the objective of \eqref{eq:finite} and collecting terms involving $x_j$ gives 
\[\label{eq:virtual} s(x_j, \theta_j)\mu(\theta_j) - \big(u(x_j, \theta_{j+1}) -  u(x_j, \theta_{j})\big)\sum_{i > j} \mu(\theta_i) \,.\tag{A.1} \]
Now consider the terms involving $x_{j^*}$. Because $o < j^* < g$, there is no IC constraint pointing toward $j^*$ by \Cref{claim:transfer}. Therefore, there is only one such term: 
\[s(x_{j^*}, \theta_{j^*}) \mu(\theta_{j^*})\,.\]
Note that $x_j \in \X$ is feasible to assign to $\theta_{j^*}$. Moreover, since $x_j \leq x_o$, doing so maintains the form of $\mathbf{T}$ by \Cref{lem:form}, and thus generates a payoff also according to the above formula. The fact that $x$ is optimal then implies 
\[s(\hat{x}, \hat{\theta}) \geq s(x_j, \hat{\theta})\,;\]
that is, 
\[ s(x_j, \hat{\theta}) - s(\hat{x}, \hat{\theta})\leq 0\,.\]
Because $x_j > \hat{x}$ and $ \theta_j < \hat{\theta}$,
by the surplus condition on $s$, 
\[\label{eq:surplus} s(x_j, \theta_j) - s(\hat{x}, \theta_j)\leq 0 \,.\tag{A.2}\]
Moreover, because $x_j > \hat{x}$, by the strict increasing differences property of $u$, 
\[\label{eq:util}  u(x_j, \theta_{j+1}) -  u(x_j, \theta_{j}) > u(\hat{x}, \theta_{j+1}) -  u(\hat{x}, \theta_{j})\,. \tag{A.3}\]
Combining \eqref{eq:surplus} and \eqref{eq:util} gives
\[s(x_j, \theta_j)\mu(\theta_j) - \big(u(x_j, \theta_{j+1}) -  u(x_j, \theta_{j})\big)\sum_{i > j} \mu(\theta_i) \leq s(\hat{x}, \theta_j)\mu(\theta_j) - \big(u(\hat{x}, \theta_{j+1}) -  u(\hat{x}, \theta_{j})\big)\sum_{i > j} \mu(\theta_i)\,, \]
proving that the part of the objective involving $x_j$ increases.

Because this holds for all $j \in\{ k, k+1, \dots, o-1 \}$, to conclude our proof, it remains to show that the part of the objective involving $x_o$ strictly increases. Plugging $\varphi(S, \,\cdot\,)$ into \eqref{eq:finite} and collecting terms involving $x_o$ gives 
\[s(x_o, \theta_o) \mu(\theta_o) -  \sum_{i=o+1}^g \mu(\theta_i)\big (u(x_o, \theta_{i})- u(x_o, \theta_o) \big ) - \big (u(x_o, \theta_g) - u(x_o, \theta_o) \big)\sum_{i>g} \mu(\theta_i)\,. \]
By the same argument as the previous case, we have 
\[s(x_o, \theta_o) \leq s(\hat{x}, \theta_o)\,.\]
For any $i > o$, by the strict increasing differences property of $u$, 
\[u(x_o, \theta_{i})- u(x_o, \theta_o)  > u(\hat{x}, \theta_{i})- u(\hat{x}, \theta_o)\,. \]
Together they imply
\begin{align*}
s(x_o, \theta_o) \mu(\theta_o) -  \sum_{i=o+1}^g \mu(\theta_i)\big (u(x_o, \theta_{i})- u(x_o, \theta_o) \big ) - \big (u(x_o, \theta_g) - u(x_o, \theta_o) \big)\sum_{i>g} \mu(\theta_i) \\
< s(\hat{x}, \theta_o) \mu(\theta_o) -  \sum_{i=o+1}^g \mu(\theta_i)\big (u(\hat{x}, \theta_{i})- u(\hat{x}, \theta_o) \big ) - \big (u(\hat{x}, \theta_g) - u(\hat{x}, \theta_o) \big)\sum_{i>g} \mu(\theta_i)\,,
\end{align*}
where the strict inequality also uses that $\mu$ has full support. 
\vspace{0.5cm}
\begin{proof}[Proof of \Cref{lem:form}]
Fix any subset $K \subset \{1, \dots, n\}$, any index $k \not\in K$, and any allocation rule $a \in \R^n$. We claim that if $a_k = a_{\max\{j \in K: j < k\}}$, then
\[\varphi(K \cup \{k\}, a) = \varphi(K, a)\,.\]
Let $i, i'$ be the two consecutive indices in $K$ such that $i < k < i'$. Note that for any $j \leq k$, \[(\varphi(K \cup \{k\}, a))_j = (\varphi(K, a))_j\,,\]
since $(K \cup \{k\}) \cap \{j': j' < j\} = K \cap \{j': j' < j\} $. For any $k < j \leq i'$, we can write 
\[u(a_k, \theta_j) - u(a_k, \theta_{k}) + u(a_i, \theta_k) - u(a_i, \theta_{i}) =  u(a_k, \theta_j)   - u(a_i, \theta_{i}) =   u(a_i, \theta_j)   - u(a_i, \theta_{i})\,,\]
and $i = \max\{j' \in K: j' < j\}$.  Thus $(\varphi(K \cup \{k\}, a))_j = (\varphi(K, a))_j$ for any $k < j \leq i'$. Also, the fact that the above holds for $j = i'$ implies that $(\varphi(K \cup \{k\}, a))_j = (\varphi(K, a))_j$ for any $j > i'$.  

Now, write $S_b = S_a \cup \{k_1, \dots, k_m\}$ with $k_1 \leq \dots \leq k_m$. By assumption, $b_{k_m} = \dots = b_{k_1} = b_{\max \{j\in S_a: j < k_1\}}$. Then, by the definition of $S_b$, for all $q = 2, \dots, m$, we have 
\[b_{k_q} = b_{\max\{j \in S_a \cup \{k_1, \dots, k_{q-1}\}: j < k_q\}}\,. \]
Thus, we can repeatedly apply the result from the previous paragraph and obtain 
\[\mathbf{T}[b] = \varphi(S_a \cup \{k_1, \dots, k_m\}, b) = \varphi(S_a \cup \{k_1, \dots, k_{m-1}\}, b) = \dots =  \varphi(S_a \cup \{k_1\}, b)  = \varphi(S_a, b)\,,\]
which proves the lemma. 
\end{proof}

\subsubsection{Completion of Proof of \texorpdfstring{\Cref{thm:dbind}}{}}\label{subsubsec:completion}

We complete the proof of \Cref{thm:dbind} by extending the finite-case result (see \Cref{subsec:dst}) to the general type space $\Theta$ via approximation. To ease the reading, we first sketch the general argument, prove the technical lemmas, and then finish the formal proof. 

\paragraph{Sketch of the Argument.}\hspace{-2mm}We first give a sketch of the argument here and then provide the formal proof. Let $\mu \in \Delta(\Theta)$ denote the distribution on $\Theta$. Recall that $\eqref{eq:1d}^\dagger$ denotes the version of program $\eqref{eq:1d}$ with all IC constraints (both downward and upward). Let $V(\Theta, \mu)$ denote the optimal value of $\eqref{eq:1d}^\dagger$ given $(\Theta, \mu)$. We show that $V(\Theta, \mu)$ equals the optimal value of $\eqref{eq:1d}$. Suppose, for contradiction, there exists some $(\hat{x}, \hat{t})$ feasible for $\eqref{eq:1d}$ such that  
\[\label{eq:contrad} V(\Theta, \mu) < \E^\mu[v(\hat{x}(\theta), \theta) + \hat{t}(\theta)] \,.\tag{A.4}\]
We first construct an appropriate sequence $\{(\Theta^{(n)}, \mu^{(n)})\}$ approximating $(\Theta, \mu)$. 
\begin{lemma}[Approximation]\label{lem:app}
Suppose $v(x, \theta)$ is Lipschitz continuous on $\X \times \Theta$. Then, there exists a sequence $\{(\Theta^{(n)}, \mu^{(n)})\}$ with $\Theta^{(n)}\subseteq \Theta$ finite and $\mu^{(n)}\in \Delta(\Theta^{(n)})$ full support such that 
\begin{itemize}
    \item[(i)]  $\mu^{(n)}  \rightarrow_w \mu$\,;
    \item[(ii)] $\displaystyle \limsup_{n \rightarrow \infty} V(\Theta^{(n)}, \mu^{(n)}) \leq V(\Theta, \mu) $\,.
\end{itemize}
\end{lemma}

Suppose for a moment that $\hat{x}, \hat{t}$ are continuous on $\Theta$ and $v$ is Lipschitz continuous. Note that $(\hat{x}, \hat{t})$ restricted to $\Theta^{(n)}$ is a feasible solution to the finite-type version of \eqref{eq:1d} with $(\Theta^{(n)}, \mu^{(n)})$. By Step 2, we have $V(\Theta^{(n)}, \mu^{(n)}) \geq \E^{\mu^{(n)}}[v(\hat{x}(\theta), \theta) + \hat{t}(\theta)]$. Because $v(\hat{x}(\theta), \theta) + \hat{t}(\theta)$ is a bounded continuous function on $\Theta$, using \Cref{lem:app} and taking limits on both sides of the above, we have
\[V(\Theta, \mu) \geq \limsup_{n \rightarrow \infty} V(\Theta^{(n)}, \mu^{(n)}) \geq  \limsup_{n \rightarrow \infty} \E^{\mu^{(n)}}[v(\hat{x}(\theta), \theta) + \hat{t}(\theta)] = \E^\mu[v(\hat{x}(\theta), \theta) + \hat{t}(\theta)]\,, \]
contradicting \eqref{eq:contrad}. However, the situation is more delicate in general. The actual proof relies on the Stone–Weierstrass theorem and Lusin’s theorem. 

Finally, to conclude \Cref{thm:dbind}, it suffices to show the existence of an optimal solution to the full IC program $\eqref{eq:1d}^\dagger$. Even though this is a standard one-dimensional problem, the existence result appears to be new at this level of generality. 
\begin{lemma}[Existence]\label{lem:exist}
Any one-dimensional screening problem has a solution.
\end{lemma}
The proof proceeds by showing the space of IC and IR mechanisms is sequentially compact in the product topology. The argument uses a generalized version of Helly's selection theorem from \citet{fuchino1999theorem}. 

\paragraph{Proofs of Approximation and Existence Lemmas.}\hspace{-2mm}We first prove the technical lemmas.
\begin{proof}[Proof of \Cref{lem:app}]
Without loss, let $\Theta \subseteq [0,1)$ and $0 \in \Theta$. The construction works as follows. Fix any $n \in \N$. Partition $[0,1)$ into intervals $\{[\frac{i-1}{n},\frac{i}{n})\}_{i=1, \dots, n}$.  Let 
\[I = \Big \{i:  \mu([\frac{i-1}{n},\frac{i}{n})) >0 \Big\}\,.\]
For any $i \in I$, let 
\[\theta^{(n)}_i = \min \Big \{ [\frac{i-1}{n},\frac{i}{n} ) \cap \Theta \Big \}\,.\]
(The minimum is attained since $\Theta$ is compact.) For notational convenience, we reindex $i$ so that it runs over from $1$ to $|I|$. Let 
\[\Theta^{(n)} = \{\theta^{(n)}_i\}_{i \in I} \,;\]
\[\mu^{(n)}(\theta^{(n)}_i) = \mu([\theta^{(n)}_i, \theta^{(n)}_{i+1}))\,.\]
We have $\Theta^{(n)} \subseteq \Theta$ finite and $\mu^{(n)} \in \Delta(\Theta^{(n)})$ full support. Note that \[\label{eq:dist} \mu(\{\theta \in \Theta: \theta \in [\theta^{(n)}_i, \theta^{(n)}_{i+1}) \text{ and } |\theta - \theta^{(n)}_i| > \frac{1}{n}  \}) = 0 \,.\tag{A.5}\]

We first show property $(ii)$ in the statement. Recall that for this lemma we assume $v$ is Lipschitz continuous on $\X \times \Theta$. Then, there exists some constant $K > 0$ such that for any $\theta, \theta' \in \Theta$,  
\[\label{eq:Lip}\max_{x \in \X} |v(x, \theta') - v(x, \theta)| 
\leq K |\theta' - \theta|\,. \tag{A.6}\]
Let $(x^{(n)}, t^{(n)})$ be any optimal solution to the full IC program $\eqref{eq:1d}^\dagger$ with $(\Theta^{(n)}, \mu^{(n)})$. Let $\bar{x}^{(n)}$ be the extension of $x^{(n)}$ to the right: 
\[\bar{x}^{(n)}(\theta) = x^{(n)}(\theta^{(n)}_i) \quad \text{for all } \theta \in [\theta^{(n)}_i, \theta^{(n)}_{i+1})\,.\]
Note that $\bar{x}^{(n)}$ is a monotonic function on $[0, 1)$. Define $\bar{t}^{(n)}$ in the same way. We claim $(\bar{x}^{(n)}, \bar{t}^{(n)})$, when restricted to $\Theta$, is a feasible solution to $\eqref{eq:1d}^\dagger$ with $(\Theta, \mu)$. To see this, offer the menu $\{(x^{(n)}_i, t^{(n)}_i)\}_{i \in I}$ to all types in $\Theta$. Type $\theta^{(n)}_{i+1}$ is indifferent between $(x^{(n)}_{i+1}, t^{(n)}_{i+1})$ and $(x^{(n)}_i, t^{(n)}_i)$. Type  $\theta^{(n)}_i$ finds $(x^{(n)}_i, t^{(n)}_i)$ optimal. Therefore, any type $\theta$ between $\theta^{(n)}_i$ and $\theta^{(n)}_{i+1}$ finds $(x^{(n)}_i, t^{(n)}_i)$ optimal since $u$ has strict increasing differences. By  construction,
\[V(\Theta^{(n)}, \mu^{(n)}) = \E^{\mu^{(n)}}[v(\bar{x}^{(n)}(\theta), \theta) + \bar{t}^{(n)}(\theta)]\,. \]
Since $(\bar{x}, \bar{t})$ is feasible to  $\eqref{eq:1d}^\dagger$ with $(\Theta, \mu)$, we have 
\[V(\Theta, \mu) \geq \E^{\mu}[v(\bar{x}^{(n)}(\theta), \theta) + \bar{t}^{(n)}(\theta)] \,.\]
Because $\bar{x}^{(n)}, \bar{t}^{(n)}$ are constant over each interval $[\theta^{(n)}_i, \theta^{(n)}_{i+1})$, by \eqref{eq:dist} and \eqref{eq:Lip}, we have 
\begin{align*} \label{eq:diff}
  &\Big | \E^{\mu}[v(\bar{x}^{(n)}(\theta), \theta) + \bar{t}^{(n)}(\theta)] - \E^{\mu^{(n)}}[v(\bar{x}^{(n)}(\theta), \theta) + \bar{t}^{(n)}(\theta)] \Big | \\
    &= \Big |\int v(\bar{x}^{(n)}(\theta), \theta) \d\mu -  \int v(\bar{x}^{(n)}(\theta), \theta) \d\mu^{(n)} \Big| \\
    &\leq \sum_{i\in I} \mu^{(n)}(\theta^{(n)}_i) \sup_{\theta \in  [\theta^{(n)}_i, \theta^{(n)}_i+\frac{1}{n}] \cap \Theta } \Big \{ \max_{x\in \X} \big | v(x, \theta) - v(x,\theta^{(n)}_{i})\big | \Big \}\\
    &\leq \sum_{i\in I} \mu^{(n)}(\theta^{(n)}_i) \frac{K}{n} = \frac{K}{n} \,. \tag{A.7}
\end{align*}
Then, it follows that
\[V(\Theta, \mu) \geq V(\Theta^{(n)}, \mu^{(n)})  - \frac{K}{n}\,.\]
Taking $\limsup$ on both sides gives property $(ii)$ in the statement. 

We now show property $(i)$ in the statement. It suffices to prove the weak convergence in $\Delta([0,1])$. Let $F$, $F^{(n)}$ be the CDFs of $\mu$, $\mu^{(n)}$. We have $ F^{(n)}(1) = F(1)  = 1$. Fix any $\theta \in [0, 1)$. Note that $\mu^{(n)} \preceq \mu$ in the stochastic dominance order, and hence 
\[\label{eq:fosd}
 F^{(n)}(\theta) \geq F(\theta) \,.\tag{A.8} \]
Let $i$ be such that $[\theta^{(n)}_{i}, \theta^{(n)}_{i+1}) \ni \theta$. Note that 
\[F^{(n)}(\theta) = \mu^{(n)}([0, \theta]) \leq \mu^{(n)}([0, \theta^{(n)}_{i+1}))= \mu([0, \theta^{(n)}_{i+1}))\,. \]
If $\theta + \frac{1}{n} \geq \theta^{(n)}_{i+1}$, then we have  
\[\mu([0, \theta^{(n)}_{i+1}))\leq F(\theta + \frac{1}{n}) \,.\]
Otherwise, since $\theta + \frac{1}{n} \geq \theta_{i}^{(n)} + \frac{1}{n}$, we have $\mu([\theta + \frac{1}{n}, \theta^{(n)}_{i+1})) = 0$. Thus,
\[\mu([0, \theta^{(n)}_{i+1})) = \mu([0, \theta + \frac{1}{n}))\leq F(\theta + \frac{1}{n})\,. \]
Hence, in either case, we have 
\[\label{eq:rev} F^{(n)}(\theta) \leq F(\theta + \frac{1}{n}) \,.\tag{A.9}\]
Using \eqref{eq:fosd}, \eqref{eq:rev}, and that $F$ is right-continuous, we have 
\[F(\theta) \leq  \lim_{n \rightarrow \infty }F^{(n)}(\theta) \leq \lim_{n \rightarrow \infty}  F(\theta + \frac{1}{n}) = F(\theta)\,.\]
Therefore, $F^{(n)}$ converges to $F$ pointwise, and hence $\mu^{(n)}\rightarrow_w \mu$. 
\end{proof}

\vspace{0.5cm}
\begin{proof}[Proof of \Cref{lem:exist}]
Recall $\M(\Theta)$ is the set of IC and IR mechanisms for the one-dimensional type space $\Theta$. We want to show the following program has a solution:
\[\sup_{(x, t) \in \M(\Theta)}\E [v(x(\theta), \theta) + t(\theta) ] \,.\]

We first show that it is without loss to restrict the range of $t$ to some interval $[-K, K]$ for $K$ large enough. By the IR constraints, we have $t(\theta) \leq \max_{x, \theta}|u(x, \theta)|$. By the IC constraints, for any $\theta, \theta'$, we have
\[|t(\theta) - t(\theta')| \leq 2\max_{x,\theta}|u(x,\theta)|\,.\]
Hence, for all $\theta$, 
\[\label{eq:tbound} t(\theta)\geq -3\max_{x,\theta}|u(x,\theta)| -2\max_{x,\theta}|v(x,\theta)| \,, \tag{A.10}\]
because if the above is violated at any type $\theta$, the principal gets strictly less than \[-\max_{x,\theta}|u(x,\theta)| - \max_{x,\theta}|v(x,\theta)|\]
but that can be easily obtained by offering a single option. Thus, the claim holds for $K = 3\max_{x,\theta}|u(x,\theta)|+2\max_{x,\theta}|v(x,\theta)|$.

Then, $\M(\Theta) \subseteq \X^\Theta \times[-K, K]^{\Theta}$ (with the product topology); we use the notation $\X^\Theta := \times_{\theta \in \Theta} \X$. By the dominated convergence theorem, the objective is sequentially continuous on $\M(\Theta)$. It is clear that $\M(\Theta)$ is nonempty. The existence result follows once we show $\M(\Theta)$ is sequentially compact. Fix any sequence $\{(x^{(n)},t^{(n)})\}_{n}$ in $\M(\Theta)$. Let 
\[U^{(n)}(\theta) = u(x^{(n)}(\theta), \theta) - t^{(n)}(\theta)\]
be the equilibrium payoff of type $\theta$. For any $\hat{\theta} < \theta$, by IC$[\theta \rightarrow \hat{\theta}]$, we have
\[U^{(n)}(\hat{\theta}) =  u(x^{(n)}(\hat{\theta}), \hat{\theta}) - t^{(n)}(\hat{\theta}) \leq u(x^{(n)}(\hat{\theta}), \theta) - t^{(n)}(\hat{\theta}) \leq U^{(n)}(\theta)\,. \]
Therefore, $U^{(n)} \in [-K, K]^{\Theta}$ is a monotone function (increase $K$ if necessary). Since $u$ has strict increasing differences, $x^{(n)} \in \X^\Theta$ is also a monotone function. Note that $\Theta, \X \subset  \R$ are linearly ordered and sequentially compact sets. By Helly's selection theorem for monotone functions on linearly ordered sets (\citealt{fuchino1999theorem}, Theorem 7), there exists a subsequence $\{x^{(n_k)}\}$ that converges pointwise. Applying the same theorem again on $\{U^{(n_k)}\}$, we obtain a subsubsequence $\{U^{(n_{k_l})}\}$ that converges pointwise. Therefore, 
\[t^{(n_{k_l})}(\theta) = u(x^{(n_{k_l})}(\theta), \theta) - U^{(n_{k_l})}(\theta)\]
also converges pointwise by continuity of $u$. Thus, there exists some $(x^*, t^*) \in \X^\Theta \times [-K, K]^{\Theta}$ such that 
\[(x^{(n_{k_l})}, t^{(n_{k_l})}) \rightarrow (x^*, t^*)\]
in the product topology. Being the pointwise limit of measurable real-valued functions, $x^*$ is measurable; so is $t^*$. Moreover, for any $\theta, \hat{\theta} \in \Theta$, 
\begin{align*}
u(x^*(\theta), \theta) - t^*(\theta) &= \lim_{l \rightarrow \infty}\big(u(x^{(n_{k_l})}(\theta), \theta) - t^{(n_{k_l})}(\theta) \big)\\    
&\geq \lim_{l \rightarrow \infty}\big(u(x^{(n_{k_l})}(\hat{\theta}), \theta) - t^{(n_{k_l})}(\hat{\theta}) \big) =  u(x^*(\hat{\theta}), \theta) - t^*(\hat{\theta})
\end{align*}
by continuity of $u$ and that $(x^{(n)}, t^{(n)})\in \M(\Theta)$ for all $n$. Therefore, $(x^*, t^*)$ satisfies all IC constraints. Similarly,  $(x^*, t^*)$ satisfies all IR constraints. So $(x^*, t^*) \in \M(\Theta)$, and hence $\M(\Theta)$ is sequentially compact. 
\end{proof}

\paragraph{Completion of the Proof of \Cref{thm:dbind}.}\hspace{-2mm}Recall that we want to show the optimal value of \eqref{eq:1d} equals $V(\Theta, \mu)$. We first show it for Lipschitz continuous $v$, and then extend it to all continuous $v$. Without loss, we assume $0 \in \Theta \subseteq [0, 1]$. Suppose for contradiction that there exist some $(\hat{x}, \hat{t})$ feasible for \eqref{eq:1d} and some $\epsilon > 0$ such that
\[\label{eq:upperbound} V(\Theta, \mu) + \epsilon \leq \E^\mu[v(\hat{x}(\theta), \theta) + \hat{t}(\theta)]\,. \tag{A.11}\]
Let $\bar{S} = 3\max_{x,\theta} |u(x,\theta)| + 3\max_{x,\theta}|v(x,\theta)|$. By Lusin's theorem (see e.g. \citealt{Aliprantis2006}, Theorem 12.8), there exists a compact set $\tilde{\Theta} \subseteq \Theta$ such that $\hat{x}, \hat{t}$ are continuous on $\tilde{\Theta}$ and $\alpha := \mu(\Theta \backslash \tilde{\Theta}) < \epsilon/(3\bar{S}) $. Since $\tilde{\Theta}$ is compact, $\underline{\tilde{\Theta}} := \min\{\tilde{\Theta}\}$ is attained. If $\underline{\tilde{\Theta}} > 0$, we augment $\tilde{\Theta}$ by adding $\theta = 0$. Since $\{0\}$ is a singleton disjoint from the compact set $\tilde{\Theta}$, we have $\hat{x}, \hat{t}$ continuous on the augmented set as well. Since $(\hat{x}, \hat{t})$ is IR, $\hat{t}(\theta) \leq \max_{x,\theta} |u(x,\theta)|$, and hence 
\[\label{eq:app} \E^\mu[v(\hat{x}(\theta), \theta) + \hat{t}(\theta)] \leq  (1-\alpha) \E^{\tilde{\mu}}[v(\hat{x}(\theta), \theta) + \hat{t}(\theta)] + \alpha \bar{S}\,, \tag{A.12}\]
where $\tilde{\mu}$ is the distribution of $\theta$ conditional on $\theta \in \tilde{\Theta}$. We pick an approximation sequence $\{(\Theta^{(n)}, \mu^{(n)})\}$ for $(\tilde{\Theta}, \tilde{\mu})$ according to \Cref{lem:app}. By \eqref{eq:diff}, for all $n$ large enough, we have
\[\label{eq:app2} \E^{\mu^{(n)}}[v(\bar{x}^{(n)}(\theta), \theta) + \bar{t}^{(n)}(\theta)] -   \frac{\epsilon}{3(1-\alpha)}\leq \E^{\tilde{\mu}}[v(\bar{x}^{(n)}(\theta), \theta) + \bar{t}^{(n)}(\theta)] \,, \tag{A.13}\]
where $(x^{(n)}, t^{(n)})$ is an optimal solution to the full IC problem $\eqref{eq:1d}^\dagger$ with $(\Theta^{(n)}, \mu^{(n)})$, and $(\bar{x}^{(n)}, \bar{t}^{(n)})$ is the extension of $(x^{(n)}, t^{(n)})$ to the right, as defined in the proof of \Cref{lem:app}. As in the proof of \Cref{lem:app}, $(\bar{x}^{(n)}, \bar{t}^{(n)})$ satisfies all IC and IR constraints for type space $\Theta$. As in the proof of \Cref{lem:exist},  \eqref{eq:tbound} then holds for $\bar{t}^{(n)}$. By feasibility, \eqref{eq:tbound}, and  \eqref{eq:app2}, 
\begin{align*}
    V(\Theta, \mu) &\geq \E^{\mu}[v(\bar{x}^{(n)}(\theta), \theta) + \bar{t}^{(n)}(\theta) ] \\
    &\geq (1-\alpha)\E^{\tilde{\mu}}[v(\bar{x}^{(n)}(\theta), \theta) + \bar{t}^{(n)}(\theta) ] -\alpha \bar{S}\\
     &\geq (1-\alpha)\E^{\mu^{(n)}}[v(\bar{x}^{(n)}(\theta), \theta) + \bar{t}^{(n)}(\theta) ]  -\frac{2}{3}\epsilon \\ 
     &= (1-\alpha)V(\Theta^{(n)}, \mu^{(n)})  -\frac{2}{3}\epsilon \\ 
     &\geq (1-\alpha)\E^{\mu^{(n)}}[v(\hat{x}( \theta), \theta) + \hat{t}(\theta) ]  -\frac{2}{3}\epsilon \,.
\end{align*}
In the last inequality, we have used that $(\hat{x}, \hat{t})$ is a downward IC and IR mechanism for $(\Theta^{(n)}, \mu^{(n)})$ and that \Cref{thm:dbind} holds for finite type spaces (see \Cref{subsec:dst}). Because $\hat{x}, \hat{t}$ are bounded and continuous on $\tilde{\Theta}$, and $v$ is continuous on the compact space $\X \times \Theta$, we have $v(\hat{x}( \theta), \theta) + \hat{t}(\theta)$ is bounded and continuous on $\tilde{\Theta}$. But, since $\mu^{(n)} \rightarrow_w \tilde{\mu}$ in $\Delta(\tilde{\Theta})$, taking limits on both sides of the above and using \eqref{eq:app}, we see that 
\begin{align*}
V(\Theta, \mu) &\geq (1-\alpha)\E^{\tilde{\mu}}[v(\hat{x}( \theta), \theta) + \hat{t}(\theta) ]  -\frac{2}{3}\epsilon    \\
&\geq \E^{\mu}[v(\hat{x}( \theta), \theta) + \hat{t}(\theta) ] - \alpha \bar{S} -\frac{2}{3} \epsilon\\
& > \E^{\mu}[v(\hat{x}( \theta), \theta) + \hat{t}(\theta) ] - \epsilon\,,
\end{align*}
which is a direct contradiction to \eqref{eq:upperbound}.

Now we let $v$ be any continuous function on $\X \times \Theta$. Since $\X \times \Theta$ is compact, as a consequence of the Stone–Weierstrass theorem (see e.g. \citealt{Aliprantis2006}, Theorem 9.13), the set of Lipschitz continuous real-valued functions on $\X \times \Theta$ is dense in the space of continuous functions on  $\X \times \Theta$ (with the sup norm). Therefore, there exists a sequence of Lipschitz continuous functions $\{v_{k}\}$ converging uniformly to $v$. Passing to a subsequence if necessary, we may assume that for all $k$, 
\[\sup_{x \in \X, \theta \in \Theta} |v_{k}(x, \theta) - v(x, \theta)| < \frac{1}{k}\,.\]
Using the above and the earlier result applied to $v_k$, we have for all $k$, 
\begin{align*}
\sup_{(x,t) \in \tilde{\M}(\Theta)} \E\big [v(x(\theta), \theta) + t(\theta) \big ] - \frac{1}{k}  &\leq \sup_{(x,t) \in \tilde{\M}(\Theta)}\E\big [v_{k}(x(\theta), \theta) + t(\theta) \big ] \\
&\leq \sup_{(x,t) \in \M(\Theta)}\E\big [v_{k}(x(\theta), \theta) + t(\theta) \big ] \leq \sup_{(x,t) \in \M(\Theta)}\E\big [v(x(\theta), \theta) + t(\theta) \big ] + \frac{1}{k}\,.   
\end{align*}
Taking $k \rightarrow \infty$ then gives the desired inequality. 

Invoking the existence result of \Cref{lem:exist}, we conclude the proof of \Cref{thm:dbind}.

\subsection{Proof of \texorpdfstring{\Cref{thm:main}}{}}

Recall that $\Theta_\varepsilon=\{(\theta^A,\theta^B): \theta^B = h(\theta^A; \varepsilon), \theta^A \in \Theta^A\}$ is the decomposed monotonic path given a realization $\varepsilon$. For any type space $\Theta$, recall that $\mathcal{M}(\Theta)$ is the set of IC and IR mechanisms. Note that the principal's optimal payoff can be bounded above as follows: 
\begin{align*}\label{eq:split}
\hspace{-10mm}    \sup_{(x,y,t) \in \mathcal{M}(\Theta)} \E\Big[v^A(x(\theta), \theta^A) + v^B(y(\theta), \theta^B)+t(\theta)\Big]  &\leq \E_\varepsilon\Big[ \sup_{(x,y,t) \in \mathcal{M}(\Theta)} \E\big [v^A(x(\theta), \theta^A) + v^B(y(\theta), \theta^B)+t(\theta) \mid \varepsilon \big ] \Big ] \\
     &\leq \E_\varepsilon\Big[ \sup_{(x,y,t) \in \mathcal{M}(\Theta_\varepsilon)} \E\big [v^A(x(\theta), \theta^A) + v^B(y(\theta), \theta^B) +t(\theta) \mid \varepsilon \big ] \Big ]  \,.
\end{align*}
Because $\varepsilon$ is independent of $\theta^A$, the inner expectation integrates with respect to the same marginal distribution of $\theta^A$ regardless of the realization of $\varepsilon$. 

By the proof in \Cref{sec:proof}, note that for all realizations of $\varepsilon$, 
\[ \sup_{(x,y,t) \in \mathcal{M}(\Theta_\varepsilon)} \E\big [v^A(x(\theta), \theta^A) + v^B(y(\theta), \theta^B) +t(\theta)  \mid \varepsilon \big ] \]
can be attained by a single mechanism in $\mathcal{M}(\Theta)$ that involves no costly screening and does not depend on $\varepsilon$. The first part of \Cref{thm:main} follows immediately.

Now, for the second part of \Cref{thm:main} (uniqueness part), note that if the instruments are strictly costly, then by \Cref{lem:dom} and \Cref{thm:dbind}, for all realizations of $\varepsilon$, all optimal solutions to the above problem satisfy $\P(y(\theta) = y_0 \mid \varepsilon) = 1$. Now, if a mechanism $(x, y, t)$ in $\mathcal{M}(\Theta)$ has $y(\theta) \neq y_0$ for a positive measure of $\theta$, then we have $\P(y(\theta) = y_0 \mid \varepsilon) < 1$ for a positive measure of $\varepsilon$. Hence, if the instruments are strictly costly, then $(x, y, t)$ is strictly dominated by any optimal mechanism involving no costly screening, and thus the second part of \Cref{thm:main} follows.

\subsection{Proof of \texorpdfstring{\Cref{prop:converse}}{}}
Without loss of generality, we may assume $i = 1$. Let $v^A = v^B = 0$. For convenience, we write $\theta^A$ as $\theta^0$.  Since $\theta^1$ is stochastically nonincreasing in $\theta^0$, we have that $\theta^1$ and $-\theta^0$ are positively upper orthant dependent (\citealt{muller2002comparison}, pp. 121-125), and hence 
\[\P(-\theta^0 > -q^0, \theta^1 > q^1) \geq \P(-\theta^0 > -q^0) \P(\theta^1 > q^1) \]
for all $q^0 \in \Theta^0, q^1 \in \Theta^1$. Since $\theta^0$ and $\theta^1$ are not independent, there exist $q^0, q^1$ such that 
\[\P(-\theta^0 > -q^0, \theta^1 > q^1) > \P(-\theta^0 > -q^0) \P(\theta^1 > q^1) \,.\]
That is, 
\[\P(\theta^0 < q^0, \theta^1 > q^1) > \P(\theta^0 < q^0) \P(\theta^1 > q^1) \,.\]
Clearly, for the above to hold, we must have $\P(\theta^0 < q^0) \in (0, 1)$ and $\P(\theta^1 > q^1) \in (0, 1)$. Let 
\[\mu^0 := \P(\theta^0 > q^0),\quad  \mu^1 := \P(\theta^1 > q^1)\,.\]
As the type distribution is absolutely continuous, we have $\mu^0 \in (0, 1)$ and we can write  
\[\P(\theta^0 > q^0, \theta^1 < q^1) = \mu^0 - \mu^1 + \P(\theta^0 < q^0, \theta^1 > q^1) > \mu^0 - \mu^1 + (1-\mu^0) \mu^1 = (1-\mu^1) \mu^0\,.\]

Define 
\[
f(\theta^0) = \begin{cases}
1 & \text{if $\theta^0 \leq q^0$} \\ 
\frac{1}{\mu^0} & \text{if $\theta^0 > q^0$}
\end{cases}\,, \qquad g(\theta^1) = \begin{cases}
-\frac{1}{\mu^0} & \text{if $\theta^1 \leq q^1$} \\ 
-\epsilon & \text{if $\theta^1 > q^1$}
\end{cases}\,,\]
where $\epsilon > 0$ will be determined shortly. Let $\tilde{f}$ be a continuous approximation of $f$ such that $\tilde{f}(\theta^0) = f(\theta^0)$ for all $\theta^0 \not\in (q^0 -\epsilon, q^0 + \epsilon)$. It is clear that we may select $\tilde{f}$ to be nondecreasing. Let $x_0 = \min \X$ and $\hat{x} = \max \X$. Since $|\X| > 1$, $\hat{x} \neq x_0$. Since $|\Y| > 1$, there exists some $\hat{y} \neq y_0 \in \Y$. Now let the agent's utility functions be: 
\[u^A(x, \theta^A) = \tilde{f}(\theta^0) \frac{x - x_0}{\hat{x} - x_0}, \qquad u^B(y, \theta^B) = g(\theta^1) \1_{y \neq y_0} \,. \]
Consider offering the following menu of three options: 
\[\big \{(\hat{x}, y_0, 1/\mu^0-\epsilon), (\hat{x}, \hat{y}, 1 - 2\epsilon), (x_0, y_0, 0) \big \}\,.\]
Let the agent choose among these, breaking ties in favor of the principal. This yields a payoff of  at least 
\[r(\epsilon) := (1 - 2\epsilon) \P(\theta^1 > q^1)  + ( 1/\mu^0 -\epsilon ) \P(\theta^0 \geq q^0 + \epsilon, \theta^1 \leq q^1)  \]
for the principal. Screening the productive component alone yields a payoff of at most
\[q(\epsilon) := \frac{1}{\mu^0} \P(q^0-\epsilon \leq \theta^0 \leq q^0+\epsilon ) + 1 \]
for the principal. Note that $r(\epsilon), q(\epsilon)$ are both right-continuous at $0$, and 
\[\lim_{\epsilon \downarrow 0} r(\epsilon)= \mu^1 + \frac{1}{\mu^0}\P(\theta^0 > q^0, \theta^1 < q^1) > \mu^1 + 1 - \mu^1 = 1 = \lim_{\epsilon \downarrow  0} q(\epsilon)\,.\]
Thus, there exists some $\epsilon^* > 0$ such that $r(\epsilon^*) > q(\epsilon^*)$. With this choice of $\epsilon^*$, the above construction then gives the utility functions such that the menu of three options strictly dominates any mechanism screening only the productive component. 

\subsection{Proof of \texorpdfstring{\Cref{prop:pricing}}{}}
This follows immediately from \Cref{thm:main}.

\subsection{Proof of \texorpdfstring{\Cref{prop:labor}}{}}
Let $t = -wx$ denote the expected transfer from the agent to the principal. Then we can write the agent's payoff as 
\[-C(\theta^A) x + \theta^B y - t\,,\]
and the principal's payoff as 
\[V(\theta^A) x  + t\,.\]
In this problem, there is also the constraint that $t=-wx$ must be $0$ if $x = 0$. Relax that constraint. Then this is a special case of the main model. The surplus function $(V(\theta^A) - C(\theta^A))x$ satisfies our surplus condition because $V(\theta^A) - C(\theta^A)$ is nondecreasing in $\theta^A$. Applying \Cref{thm:main} yields an optimal mechanism $(x^*, 0, t^*)$ that involves no costly screening. Note that for $(x^*, t^*)$ to be optimal for the productive component, the payment $t_0$ associated with the option $x = 0$ must be $0$ because \textit{(i)} if $t_0 > 0$ then the mechanism cannot be IR and \textit{(ii)} if $t_0 < 0$ then the principal can strictly improve upon the mechanism by increasing $t(\theta)$ uniformly by $|t_0|$ for all types $\theta$. Thus, $(x^*, 0, t^*)$ is implementable in the original problem and hence must be optimal, proving the result.  

\subsection{Proof of \texorpdfstring{\Cref{prop:pure}}{}}
The only missing detail in the proof sketch in \Cref{subsubsec:pure} is to show that the positive correlation between the MRS and the productive types suffices for our main result under multiplicative preferences. Specifically, consider the following special case of our main model (which is more general than what we needed for \Cref{subsubsec:pure} but will be useful for \Cref{app:quality}). The agent's preferences are: 
\[\theta^A u(x) + \theta^B \cdot c(y) - t \,,\]
where $u:[0, 1] \rightarrow \R$ is a continuous and strictly increasing function satisfying $u(0) = 0$, and $c: \Y \rightarrow \R_{\geq 0}^{N}$ is a bounded measurable function satisfying $c(y_0) = 0$ for some $y_0 \in \mathcal{Y}$. Note that here $\theta^A \in \R_{>0}$ and $\theta^B \in \R^N_{\leq 0}$. The principal's preferences are: 
\[-C(x) + t\,,\]
where $C$ is continuous and nondecreasing, satisfying $C(0) = 0$. That is, we assume no interdependent preferences here and that $y$ is only costly for the agent. 

We claim that if $r^B:= \frac{1}{\theta^A}\theta^B$ is stochastically nondecreasing in $\theta^A$, then there exists an optimal mechanism involving no costly screening. By \Cref{lem:decomp}, as in \Cref{subsec:decompose}, it suffices to show the case where $r^B = h(\theta^A)$ for some nondecreasing function $h: \Theta^A \rightarrow \R^N$. Thus, we may assume for all $i$, $r^i$ is deterministic conditional on $\theta^A$ and nondecreasing in $\theta^A$. Fix any $(x, y, t)$ that is IC and IR. We may assume $t \geq 0$, because the monopolist can simply replace all options with negative profits in the menu with $(0, y_0, 0)$ and weakly increase the total profit (since the monopolist's cost function does not depend on the buyer's type). Now we apply a reconstruction argument as follows. Consider the modification: $\tilde{t} = t, \tilde{y} = y_0$,  
\[\tilde{x}(\theta) = u^{-1}\big(u(x(\theta)) + \frac{1}{\theta^A} [\theta^B \cdot c(y(\theta))] \big) \,.\]
Because $u(\,\cdot\,)$ is continuous and strictly increasing with $u(0)=0$, $u^{-1}$ is defined on $[0, u(1)]$. Moreover, because $(x, y, t)$ is IR and $t \geq 0$, we have $0 \leq u(x(\theta)) + \frac{1}{\theta^A} [\theta^B \cdot c(y(\theta)) ]\leq u(x(\theta))$ for all $\theta$. So the modification is well-defined and $0 \leq \tilde{x} \leq x$ pointwise. In other words, the modified mechanism decreases the productive allocation to substitute the costly screening so that all types have the same utilities as before, \textit{assuming} truthful reporting. 

Because $C(\,\cdot\,)$ is nondecreasing, this modification increases the objective, assuming truthful reporting. It is IR by construction. Moreover, it is downward IC: for any $\hat{\theta}^A < \theta^A$, 
\begin{align*}
\theta^A u(\tilde{x}(\theta)) - \tilde{t}(\theta)  
&=  \theta^A u(x(\theta)) + \theta^B \cdot c(y(\theta)) - t(\theta)\\
&\geq  \theta^A u(x(\hat{\theta})) + \theta^B \cdot c(y(\hat{\theta})) - t(\hat{\theta})\\
&= \theta^A \big( u(x(\hat{\theta})) + r^B \cdot c(y(\hat{\theta})) \big)  - t(\hat{\theta}) \\
&\geq \theta^A \big( u(x(\hat{\theta})) + \hat{r}^B \cdot c(y(\hat{\theta})) \big)  - t(\hat{\theta})  = \theta^A u(\tilde{x}(\hat{\theta})) - \tilde{t}(\hat{\theta})\,.
\end{align*}
The first inequality holds because $(x, y, t)$ is IC. The second inequality holds because $\hat{r}^B \leq r^B$ and $c \geq 0$. Invoking \Cref{thm:dbind} concludes the proof of the claim.

\setlength\bibsep{8pt}
\bibliographystyle{ecta} 
\bibliography{references}

\begin{thebibliography}{9}
\newcommand{\enquote}[1]{``#1''}
\expandafter\ifx\csname natexlab\endcsname\relax\def\natexlab#1{#1}\fi

\bibitem[\protect\citeauthoryear{Block, Savits, and Shaked}{Block
  et~al.}{1985}]{block1985concept}
\textsc{Block, H.~W., T.~H. Savits, and M.~Shaked} (1985): \enquote{A Concept
  of Negative Dependence Using Stochastic Ordering,} \emph{Statistics \&
  Probability Letters}, 3(2), 81--86.

\bibitem[\protect\citeauthoryear{Engers and Fernandez}{Engers and
  Fernandez}{1987}]{engers1987market}
\textsc{Engers, M. and L.~Fernandez} (1987): \enquote{Market Equilibrium with
  Hidden Knowledge and Self-selection,} \emph{Econometrica}, 55(2), 425--439.

\bibitem[\protect\citeauthoryear{Kamae and Krengel}{Kamae and
  Krengel}{1978}]{kamae1978stochastic}
\textsc{Kamae, T. and U.~Krengel} (1978): \enquote{Stochastic Partial
  Ordering,} \emph{Annals of Probability}, 6(6), 1044--1049.

\bibitem[\protect\citeauthoryear{Kamae, Krengel, and O'Brien}{Kamae
  et~al.}{1977}]{kamae1977stochastic}
\textsc{Kamae, T., U.~Krengel, and G.~L. O'Brien} (1977): \enquote{Stochastic
  Inequalities on Partially Ordered Spaces,} \emph{Annals of Probability},
  5(6), 899--912.

\bibitem[\protect\citeauthoryear{Karatzas and Shreve}{Karatzas and
  Shreve}{1998}]{Karatzas1998}
\textsc{Karatzas, I. and S.~Shreve} (1998): \emph{Brownian Motion and
  Stochastic Calculus {\emph{(Second Ed.)}}}, New York: Springer-Verlag.

\bibitem[\protect\citeauthoryear{M{\"u}ller and Stoyan}{M{\"u}ller and
  Stoyan}{2002}]{muller2002comparison}
\textsc{M{\"u}ller, A. and D.~Stoyan} (2002): \emph{Comparison Methods for
  Stochastic Models and Risks}, New York: Wiley.

\bibitem[\protect\citeauthoryear{Riley}{Riley}{1979}]{riley1979informational}
\textsc{Riley, J.~G.} (1979): \enquote{Informational Equilibrium,}
  \emph{Econometrica}, 47(2), 331--359.

\bibitem[\protect\citeauthoryear{Spence}{Spence}{1978}]{spence1978product}
\textsc{Spence, M.} (1978): \enquote{Product Differentiation and Performance in
  Insurance Markets,} \emph{Journal of Public Economics}, 10(3), 427--447.

\bibitem[\protect\citeauthoryear{Stantcheva}{Stantcheva}{2014}]{stantcheva2014optimal}
\textsc{Stantcheva, S.} (2014): \enquote{Optimal Income Taxation with Adverse
  Selection in the Labour Market,} \emph{Review of Economic Studies}, 81(3),
  1296--1329.

\end{thebibliography}


\begin{thebibliography}{48}
\newcommand{\enquote}[1]{``#1''}
\expandafter\ifx\csname natexlab\endcsname\relax\def\natexlab#1{#1}\fi

\bibitem[\protect\citeauthoryear{Adams and Yellen}{Adams and
  Yellen}{1976}]{adams1976commodity}
\textsc{Adams, W.~J. and J.~L. Yellen} (1976): \enquote{Commodity Bundling and
  the Burden of Monopoly,} \emph{Quarterly Journal of Economics}, 90(3),
  475--498.

\bibitem[\protect\citeauthoryear{Akbarpour, Dworczak, and Kominers}{Akbarpour
  et~al.}{2024}]{akbarpour2024redistributive}
\textsc{Akbarpour, M., P.~Dworczak, and S.~D. Kominers} (2024):
  \enquote{Redistributive Allocation Mechanisms,} \emph{Journal of Political
  Economy}, 132(6), 1831--1875.

\bibitem[\protect\citeauthoryear{Aliprantis and Border}{Aliprantis and
  Border}{2006}]{Aliprantis2006}
\textsc{Aliprantis, C.~D. and K.~C. Border} (2006): \emph{{Infinite Dimensional
  Analysis: A Hitchhiker's Guide}}, Berlin: Springer.

\bibitem[\protect\citeauthoryear{Amador and Bagwell}{Amador and
  Bagwell}{2013}]{amador2013theory}
\textsc{Amador, M. and K.~Bagwell} (2013): \enquote{The Theory of Optimal
  Delegation with an Application to Tariff Caps,} \emph{Econometrica}, 81(4),
  1541--1599.

\bibitem[\protect\citeauthoryear{Amador and Bagwell}{Amador and
  Bagwell}{2020}]{amador2020money}
---\hspace{-.1pt}---\hspace{-.1pt}--- (2020): \enquote{Money Burning in the
  Theory of Delegation,} \emph{Games and Economic Behavior}, 121, 382--412.

\bibitem[\protect\citeauthoryear{Ambrus and Egorov}{Ambrus and
  Egorov}{2017}]{ambrus2017delegation}
\textsc{Ambrus, A. and G.~Egorov} (2017): \enquote{Delegation and Nonmonetary
  Incentives,} \emph{Journal of Economic Theory}, 171, 101--135.

\bibitem[\protect\citeauthoryear{Babaioff, Immorlica, Lucier, and
  Weinberg}{Babaioff et~al.}{2014}]{babaioff2014simple}
\textsc{Babaioff, M., N.~Immorlica, B.~Lucier, and S.~M. Weinberg} (2014):
  \enquote{A Simple and Approximately Optimal Mechanism for an Additive Buyer,}
  in \emph{Proceedings of the IEEE 55th Annual Symposium on Foundations of
  Computer Science}, FOCS ’14, 21--30.

\bibitem[\protect\citeauthoryear{Banerjee}{Banerjee}{1997}]{banerjee1997theory}
\textsc{Banerjee, A.~V.} (1997): \enquote{A Theory of Misgovernance,}
  \emph{Quarterly Journal of Economics}, 112(4), 1289--1332.

\bibitem[\protect\citeauthoryear{Bergemann, Bonatti, Haupt, and
  Smolin}{Bergemann et~al.}{2022}]{bergemann2022optimality}
\textsc{Bergemann, D., A.~Bonatti, A.~Haupt, and A.~Smolin} (2022):
  \enquote{The Optimality of Upgrade Pricing,} \emph{\emph{ arXiv:2107.10323
  [econ.TH]}}.

\bibitem[\protect\citeauthoryear{Bikhchandani and Mishra}{Bikhchandani and
  Mishra}{2022}]{bikhchandani2022selling}
\textsc{Bikhchandani, S. and D.~Mishra} (2022): \enquote{Selling two identical
  objects,} \emph{Journal of Economic Theory}, 200, 105397.

\bibitem[\protect\citeauthoryear{Cai, Devanur, and Weinberg}{Cai
  et~al.}{2016}]{cai2016duality}
\textsc{Cai, Y., N.~R. Devanur, and S.~M. Weinberg} (2016): \enquote{A
  Duality-based Unified Approach to Bayesian Mechanism Design,} in
  \emph{Proceedings of the 48th Annual ACM SIGACT Symposium on Theory of
  Computing}, STOC '16, 926--939.

\bibitem[\protect\citeauthoryear{Carroll}{Carroll}{2017}]{Carroll2017}
\textsc{Carroll, G.} (2017): \enquote{{Robustness and Separation in
  Multidimensional Screening},} \emph{Econometrica}, 85(2), 453--488.

\bibitem[\protect\citeauthoryear{Che and Zhong}{Che and
  Zhong}{2024}]{che2021robustly}
\textsc{Che, Y.-K. and W.~Zhong} (2024): \enquote{Robustly Optimal Mechanisms
  for Selling Multiple Goods,} \emph{Review of Economic Studies}, 92,
  2923--2951.

\bibitem[\protect\citeauthoryear{Condorelli}{Condorelli}{2012}]{condorelli2012money}
\textsc{Condorelli, D.} (2012): \enquote{What Money Can't Buy: Efficient
  Mechanism Design with Costly Signals,} \emph{Games and Economic Behavior},
  75(2), 613--624.

\bibitem[\protect\citeauthoryear{Condorelli}{Condorelli}{2013}]{Condorelli2013}
---\hspace{-.1pt}---\hspace{-.1pt}--- (2013): \enquote{{Market and non-market
  mechanisms for the optimal allocation of scarce resources},} \emph{Games and
  Economic Behavior}, 82, 582--591.

\bibitem[\protect\citeauthoryear{Daskalakis, Deckelbaum, and Tzamos}{Daskalakis
  et~al.}{2017}]{daskalakis2017strong}
\textsc{Daskalakis, C., A.~Deckelbaum, and C.~Tzamos} (2017): \enquote{Strong
  Duality for a Multiple-Good Monopolist,} \emph{Econometrica}, 85(3),
  735--767.

\bibitem[\protect\citeauthoryear{Deb and Roesler}{Deb and
  Roesler}{2024}]{deb2023}
\textsc{Deb, R. and A.-K. Roesler} (2024): \enquote{Multi-Dimensional
  Screening: Buyer-Optimal Learning and Informational Robustness,} \emph{Review
  of Economic Studies}, 91, 2744--2770.

\bibitem[\protect\citeauthoryear{Eso and Szentes}{Eso and
  Szentes}{2007}]{Eso2007}
\textsc{Eso, P. and B.~Szentes} (2007): \enquote{{Optimal Information
  Disclosure in Auctions and the Handicap Auction},} \emph{Review of Economic
  Studies}, 74(3), 705--731.

\bibitem[\protect\citeauthoryear{Eso and Szentes}{Eso and
  Szentes}{2017}]{esHo2017dynamic}
---\hspace{-.1pt}---\hspace{-.1pt}--- (2017): \enquote{Dynamic Contracting: An
  Irrelevance Theorem,} \emph{Theoretical Economics}, 12(1), 109--139.

\bibitem[\protect\citeauthoryear{Fuchino and Plewik}{Fuchino and
  Plewik}{1999}]{fuchino1999theorem}
\textsc{Fuchino, S. and S.~Plewik} (1999): \enquote{On a Theorem of E. Helly,}
  \emph{Proceedings of the American Mathematical Society}, 127(2), 491--497.

\bibitem[\protect\citeauthoryear{Fudenberg and Tirole}{Fudenberg and
  Tirole}{1991}]{Fudenberg1991}
\textsc{Fudenberg, D. and J.~Tirole} (1991): \emph{{Game Theory}}, MIT Press.

\bibitem[\protect\citeauthoryear{Ghili}{Ghili}{2023}]{ghili2023characterization}
\textsc{Ghili, S.} (2023): \enquote{A characterization for optimal bundling of
  products with nonadditive values,} \emph{American Economic Review: Insights},
  5(3), 311--326.

\bibitem[\protect\citeauthoryear{Haghpanah and Hartline}{Haghpanah and
  Hartline}{2021}]{haghpanah2021pure}
\textsc{Haghpanah, N. and J.~Hartline} (2021): \enquote{When is Pure Bundling
  Optimal?} \emph{Review of Economic Studies}, 88(3), 1127--1156.

\bibitem[\protect\citeauthoryear{Hart and Nisan}{Hart and
  Nisan}{2017}]{hart2017approximate}
\textsc{Hart, S. and N.~Nisan} (2017): \enquote{Approximate Revenue
  Maximization with Multiple Items,} \emph{Journal of Economic Theory}, 172,
  313--347.

\bibitem[\protect\citeauthoryear{Hartline and Roughgarden}{Hartline and
  Roughgarden}{2008}]{hartline2008optimal}
\textsc{Hartline, J.~D. and T.~Roughgarden} (2008): \enquote{Optimal Mechanism
  Design and Money Burning,} in \emph{Proceedings of the 40th Annual ACM SIGACT
  Symposium on Theory of Computing}, STOC '2008, 75--84.

\bibitem[\protect\citeauthoryear{Jehiel and Moldovanu}{Jehiel and
  Moldovanu}{2001}]{jehiel2001efficient}
\textsc{Jehiel, P. and B.~Moldovanu} (2001): \enquote{Efficient design with
  interdependent valuations,} \emph{Econometrica}, 69, 1237--1259.

\bibitem[\protect\citeauthoryear{Kamae and Krengel}{Kamae and
  Krengel}{1978}]{kamae1978stochastic}
\textsc{Kamae, T. and U.~Krengel} (1978): \enquote{Stochastic Partial
  Ordering,} \emph{Annals of Probability}, 6(6), 1044--1049.

\bibitem[\protect\citeauthoryear{Klenke}{Klenke}{2013}]{klenke2013probability}
\textsc{Klenke, A.} (2013): \emph{Probability Theory: A Comprehensive Course
  \emph{(Second Ed.)}}, Germany: Springer London.

\bibitem[\protect\citeauthoryear{Manelli and Vincent}{Manelli and
  Vincent}{2006}]{manelli2006bundling}
\textsc{Manelli, A.~M. and D.~R. Vincent} (2006): \enquote{Bundling as an
  optimal selling mechanism for a multiple-good monopolist,} \emph{Journal of
  Economic Theory}, 127, 1--35.

\bibitem[\protect\citeauthoryear{McAfee and McMillan}{McAfee and
  McMillan}{1988}]{mcafee1988multidimensional}
\textsc{McAfee, R.~P. and J.~McMillan} (1988): \enquote{Multidimensional
  Incentive Compatibility and Mechanism Design,} \emph{Journal of Economic
  theory}, 46(2), 335--354.

\bibitem[\protect\citeauthoryear{Milgrom and Weber}{Milgrom and
  Weber}{1982}]{Milgrom1982}
\textsc{Milgrom, P.~R. and R.~J. Weber} (1982): \enquote{{A Theory of Auctions
  and Competitive Bidding},} \emph{Econometrica}, 50(5), 1089--1122.

\bibitem[\protect\citeauthoryear{M{\"u}ller and Stoyan}{M{\"u}ller and
  Stoyan}{2002}]{muller2002comparison}
\textsc{M{\"u}ller, A. and D.~Stoyan} (2002): \emph{Comparison Methods for
  Stochastic Models and Risks}, New York: Wiley.

\bibitem[\protect\citeauthoryear{Mussa and Rosen}{Mussa and
  Rosen}{1978}]{mussa1978monopoly}
\textsc{Mussa, M. and S.~Rosen} (1978): \enquote{Monopoly and Product Quality,}
  \emph{Journal of Economic Theory}, 18(2), 301--317.

\bibitem[\protect\citeauthoryear{Myerson}{Myerson}{1981}]{Myerson1981}
\textsc{Myerson, R.~B.} (1981): \enquote{{Optimal Auction Design},}
  \emph{Mathematics of Operations Research}, 6(1), 58–--73.

\bibitem[\protect\citeauthoryear{Ortoleva, Safonov, and Yariv}{Ortoleva
  et~al.}{2022}]{ortoleva2022cares}
\textsc{Ortoleva, P., E.~Safonov, and L.~Yariv} (2022): \enquote{Who Cares
  More? Allocation with Diverse Preference Intensities,} Tech. rep., Working
  paper.

\bibitem[\protect\citeauthoryear{Pavan, Segal, and Toikka}{Pavan
  et~al.}{2014}]{Pavan2014}
\textsc{Pavan, A., I.~Segal, and J.~Toikka} (2014): \enquote{{Dynamic Mechanism
  Design: A Myersonian Approach},} \emph{Econometrica}, 82(2), 601--653.

\bibitem[\protect\citeauthoryear{Pavlov}{Pavlov}{2011}]{pavlov2011property}
\textsc{Pavlov, G.} (2011): \enquote{A Property of Solutions to Linear Monopoly
  Problems,} \emph{The BE Journal of Theoretical Economics}, 11(1),
  0000102202193517041663.

\bibitem[\protect\citeauthoryear{Pycia}{Pycia}{2006}]{pycia2006stochastic}
\textsc{Pycia, M.} (2006): \enquote{Stochastic vs Deterministic Mechanisms in
  Multidimensional Screening,} \emph{Working paper}.

\bibitem[\protect\citeauthoryear{Riley and Zeckhauser}{Riley and
  Zeckhauser}{1983}]{riley1983optimal}
\textsc{Riley, J. and R.~Zeckhauser} (1983): \enquote{Optimal Selling
  Strategies: When to Haggle, When to Hold Firm,} \emph{Quarterly Journal of
  Economics}, 98(2), 267--289.

\bibitem[\protect\citeauthoryear{Rochet}{Rochet}{1987}]{rochet1987necessary}
\textsc{Rochet, J.-C.} (1987): \enquote{A Necessary and Sufficient Condition
  for Rationalizability in a Quasi-linear Context,} \emph{Journal of
  Mathematical Economics}, 16(2), 191--200.

\bibitem[\protect\citeauthoryear{Rochet and Stole}{Rochet and
  Stole}{2003}]{Rochet2003}
\textsc{Rochet, J.~C. and L.~A. Stole} (2003): \enquote{{The Economics of
  Multidimensional Screening},} in \emph{Advances in Economics and
  Econometrics: Theory and Applications, Eighth World Congress, Volume 1},
  Cambridge: Cambridge University Press.

\bibitem[\protect\citeauthoryear{Spence}{Spence}{1973}]{Spence1973Job}
\textsc{Spence, M.} (1973): \enquote{{Job Market Signaling},} \emph{Quarterly
  Journal of Economics}, 87(3), 355--374.

\bibitem[\protect\citeauthoryear{Stiglitz}{Stiglitz}{2002}]{Stiglitz2002}
\textsc{Stiglitz, J.~E.} (2002): \enquote{{Information and the Change in the
  Paradigm in Economics},} \emph{American Economic Review}, 92(3), 460--501.

\bibitem[\protect\citeauthoryear{Strassen}{Strassen}{1965}]{Strassen1965}
\textsc{Strassen, V.} (1965): \enquote{{The Existence of Probability Measures
  with Given Marginals},} \emph{Annals of Mathematical Statistics}, 36(2),
  423--439.

\bibitem[\protect\citeauthoryear{Thanassoulis}{Thanassoulis}{2004}]{thanassoulis2004haggling}
\textsc{Thanassoulis, J.} (2004): \enquote{Haggling over Substitutes,}
  \emph{Journal of Economic theory}, 117(2), 217--245.

\bibitem[\protect\citeauthoryear{Yang}{Yang}{2025}]{yang2023nested}
\textsc{Yang, F.} (2025): \enquote{Nested bundling,} \emph{American Economic
  Review}, 115, 2970--3013.

\bibitem[\protect\citeauthoryear{Yang, Dworczak, and Akbarpour}{Yang
  et~al.}{2024}]{yang2024comparison}
\textsc{Yang, F., P.~Dworczak, and M.~Akbarpour} (2024): \enquote{Comparison of
  Screening Devices,} \emph{Working paper}.

\bibitem[\protect\citeauthoryear{Zeckhauser}{Zeckhauser}{2021}]{zeckhauser2021strategic}
\textsc{Zeckhauser, R.} (2021): \enquote{Strategic Sorting: The Role of Ordeals
  in Health Care,} \emph{Economics \& Philosophy}, 37(1), 64--81.

\end{thebibliography}

\newpage
\section{Online Appendix}
\label{app:b}

\subsection{Absence of Surplus Condition} \label{app:example}

The following example illustrates why the surplus condition is needed for \Cref{thm:main}: 
\begin{ex}[Absence of surplus condition]\label{ex:ubind}
We consider the same setting as in \Cref{ex:poscor} (positive correlation) in \Cref{subsec:intuition} except that we increase the cost to serve the high type. In particular, suppose that it now costs $4$ to serve the high-risk type---this would violate our surplus condition since the efficient allocation is to only trade with the low-risk type. Equivalently, the agent's and the principal's utility functions on the productive component are now: 
\[u^A(x, \theta^A) = (\theta^A + 2) x\,, \quad  v^A(x, \theta^A) = - 4 \theta^A x\,.\]
As in \Cref{ex:poscor}, the costly instrument $y$ is waiting in line, for which the low-risk type has cost $1$ and the high-risk type has cost $0$. Equivalently, the agent's and the principal's utility functions on the costly component are
\[\qquad \qquad  \qquad u^B(y,\theta^B) = (\theta^B-1) y, \qquad v^B(y, \theta^B) = 0 \,, \quad \text{where $\theta^B = \theta^A$ }\,.\]
As in \Cref{ex:poscor}, the agent has positively correlated preferences because the high-risk type has a higher utility for the insurance and lower disutility for the costly action. 

However, the costly instrument can still be useful in this example. One may verify that if the principal only screens the productive component, then the best the principal can do is not to trade and get a payoff of $0$. Now, suppose the principal can use costly screening. Consider a menu of two options $\big\{(1, 0, 2),\,\, (\frac{1}{2}, \frac{1}{2}, \frac{1}{2})\big\}$, which can be interpreted as a full insurance plan with a price of $2$, and a low-coverage plan that has a price $\frac{1}{2}$ but requires some amount of waiting (this is the same menu as in \Cref{ex:poscor}).  The low-risk type, finding waiting too costly, purchases the full insurance plan. The high-risk type, finding the low-coverage plan cheap, purchases the low-coverage plan. With this menu, the principal gets a profit of $\frac{1}{2} \times 2 + \frac{1}{2}  \times (\frac{1}{2} - \frac{1}{2} \times 4) = \frac{1}{4} > 0$, strictly better off than not trading. 
\qed
\end{ex}

Compared to \Cref{ex:poscor}, the difference is that the productive component now violates the surplus condition. In the absence of the surplus condition, we can no longer ignore the upward deviation. Indeed, suppose we impose only the downward incentive constraint. Then the principal optimally sets a price of $2$ for the full insurance plan targeted at the low-risk type, and pays the high-risk type $1$ to stay out of the market. This yields a profit $\frac{1}{2} \times 2 + \frac{1}{2} \times (-1) = \frac{1}{2}$. But, of course, this is not incentive compatible: the low-risk type wants to take the payment of $1$ as well. Therefore, consistent with the intuition provided in \Cref{subsec:intuition}, in the absence of the surplus condition, the principal cannot simply ignore the possible upward deviations when screening the productive component, and hence may make use of costly screening even if the agent has positively correlated preferences.

\subsection{Nested Bundling Example}\label{app:nestedexample}

\paragraph{A parametric example for the two-item case.}\hspace{-2mm}Suppose $\Theta = [1, 1.5]$ with a uniform distribution. Consider the following valuations: 
\[v(\{1\}, \theta) = \alpha \theta^{\kappa_1}, \qquad v(\{2\}, \theta) = \theta^{\kappa_2} -1, \qquad v(\{1, 2\}, \theta) = \theta^2 \,.\]
Fix $\alpha = 0.75$ and $\kappa_2 = 2.5$. There are decreasing differences in the values of $\{1, 2\}$ versus those of $\{2\}$. Now comparing $\{1, 2\}$ and $\{1\}$, we note that there are decreasing differences in the values if $\kappa_1 \geq 2.7$ and there are increasing differences in the values if $\kappa_1 \leq 2.3$. Here, for example:
\begin{itemize}
    \item When $\kappa_1 = 2.7$, pure bundling is optimal (with a price $\approx 1$ for the bundle).
    \item When $\kappa_1 = 1.8$, selling the menu $\{\{1\}, \{1, 2\}\}$ is optimal (with a price $\approx 0.75$ for item $1$ and a price $\approx 1.05$ for the bundle).
\end{itemize}
The case of $\kappa_1 = 1.8$ follows immediately from \Cref{cor:nested} but there is no known bundling result implying this.

\subsection{Bundling and Quality Discrimination}\label{app:quality}
We show an application to a multiple-good monopoly problem allowing for both probabilistic bundling and quality discrimination. 

We follow the same notation as in \Cref{subsubsec:pure} and generalize it as follows. A monopolist sells $G$ goods. For each bundle $b$, a consumer has value $v^b$ for getting the highest quality version of the bundle with probability one.  We assume $v^{b}\leq v^{\overline{b}}$ and $v^\emptyset = 0$. The monopolist can use probabilistic bundling, captured by a bundling allocation rule $v \mapsto \alpha(v) \in \Delta(\mathcal{B})$. In addition, the monopolist can adjust the quality of each bundle, captured by a quality allocation rule $v \mapsto q(v) \in [0, 1]^{\mathcal{B}}$. A type-$v$ consumer's payoff is given by \[\sum_{b} \alpha^b q^b v^b  - t\,. \]
The monopolist incurs a cost to improve the quality of a bundle, with a payoff given by
\[-\sum_{b} \alpha^b C(q^b) + t\,,\]
where $C(\,\cdot\,)$ is a continuous, nondecreasing, and convex function on $[0,1]$ with $C(0)=0$. 

Let  $\tau = \big(\frac{v^b}{v^{\overline{b}}}\big)_{b \subset \overline{b}}$ be the profile of values for each bundle relative to the grand bundle. 

\begin{prop} \label{prop:bundle}
If $\tau$ is stochastically nondecreasing in $v^{\overline{b}}$, then an optimal mechanism exists and can be implemented by a menu of prices for different qualities of the grand bundle. 
\end{prop}

\begin{proof}[Proof of \Cref{prop:bundle}]
By convexity of $C(\,\cdot\,)$ and Jensen's inequality, we have
\begin{align*}
  \sum_{b} \alpha^b(v) C(q^b(v)) \geq  C \Big ( \sum_{b}  \alpha^b(v)  q^b(v) \Big )   \,.
\end{align*}
Therefore, it is an upper bound on the monopolist's revenue to maximize the objective \[\label{eq:aux} \E\Big[-C \Big ( \sum_{b}  \alpha^b(v)  q^b(v) \Big )  + t(v) \Big] \,.\tag{B.1}\]
For this auxiliary problem, let us also relax the constraint $\sum_b \alpha^b = 1$ to $\sum_b \alpha^b \leq 1$. Then,  because $\alpha, q$ enter both the consumer's utility and the objective in the same way, it is without loss of generality to let $q^b = 1$ for all $b$.

We now reformulate this problem into our main model. Let $\theta^A = v^{\overline{b}}$ be the value of the grand bundle. For any proper bundle $b$, let
\[\theta^b = v^b - v^{\overline{b}}\]
be the difference in values for bundle $b$ and the grand bundle $b^*$.  In words, $\theta^b$ is the negative value for getting bundle $b$ instead of $b^*$. Let $N = 2^G - 1$, and let $\theta^B = (\theta^1, \dots, \theta^{N})$ be the profile of the differences. 

We use the same mapping as in \Cref{subsubsec:pure}. We use $x: \Theta \rightarrow [0,1]$ to denote the \textit{initial} allocation of the grand bundle, and $y: \Theta \rightarrow [0,1]^{N}$ to denote the costly instruments as follows. An assignment $y^b \in [0, 1]$ represents assigning bundle $b$ with probability $y^b$ while \textit{decreasing} the probability of the grand bundle $b^*$ also by $y^b$. The consumer's payoff can be rewritten as 
\[\theta^A x +  \theta^B \cdot y - t \,.\]
For any substochastic allocation $\alpha$ (i.e. $\sum_b \alpha^b \leq 1$), we can replicate it by setting 
\[x = \sum_{b} \alpha^b\,,  \qquad y^b = \alpha^b \text{ for all } b \neq b^* \,.\]
Therefore, the auxiliary problem \eqref{eq:aux} can be further relaxed to
\[\label{eq:costly} \sup_{(x, y, t) \in \M(\Theta) } \E[-C(x(\theta)) + t(\theta)]\,. \tag{B.2}\]
For any $b = 1, \dots, N$, we have $\frac{v^b}{v^{\overline{b}}} = \frac{\theta^A + \theta^b}{\theta^A} = 1 + \frac{\theta^b}{\theta^A}$. Since $\tau$ is stochastically nondecreasing in $v^{\overline{b}}$, we have that $r^B := \frac{1}{\theta^A} \theta^B$ is stochastically nondecreasing in $\theta^A$. By the proof of \Cref{prop:pure}, we know that \eqref{eq:costly} admits an optimal mechanism involving no costly screening. Let $(x^*, 0, t^*)$ be the optimal solution to \eqref{eq:costly} that involves no costly screening.

We construct an allocation rule in the original problem as follows:
\[\text{$\alpha^{b^*} = 1$, $\alpha^b = 0$ for all $b \neq b^*$; \quad $q^{b^*} = x^*$,  $q^{b} = 0$ for all $b \neq b^*$.}\]
Because probabilities and qualities enter the consumer's utility in the same way and $(x^*, 0, t^*)$ is IC and IR,  $(\alpha, q, t^*)$ is also IC and IR. The revenue of the monopolist under $(\alpha, q, t^*)$ is \[\E[-C(q^{b^*}(\theta)) + t^*(\theta)] = \E[-C(x^*(\theta)) + t^*(\theta)]\,,\]
the optimal value of  \eqref{eq:costly}. Hence, $(\alpha, q, t^*)$ is optimal for the monopolist in the original problem; moreover, $(\alpha, q, t^*)$ screens using only the qualities of the grand bundle. 
\end{proof}

\subsection{Competitive Labor Market Screening} \label{app:comp}

Our main model assumes monopolistic screening. It delivers a prediction different from the usual perception of costly screening in competitive labor markets. In this appendix, we formulate a stylized competitive screening model consisting of two screening devices and show how competition can reverse the use of costly instruments.

There are two types of workers $\theta_H > \theta_L \geq 0$ in a perfectly competitive labor market.\footnote{This part of the setup is standard; see e.g. \citetapp{spence1978product} and \citetapp{stantcheva2014optimal}. } A type-$\theta_i$ worker incurs a cost $\psi_i(x)$ for producing $x \in [0,1]$ units of work where $\psi_i$ is a strictly increasing, continuously differentiable, and strictly convex function on $[0, 1]$ with $\psi_i(0) = 0$. A firm gets a payoff $\theta_i x$ from $x$ units of work by a type-$\theta_i$ worker. 

Suppose the marginal cost is lower for the higher type: $\psi'_H(x) < \psi'_L(x)$ for all $x \in [0,1]$. The efficient amount of production for type $\theta_i$ is $x^{e}_i := (\psi_{i}')^{-1}(\theta_i)$, assumed to be in the interior of $[0, 1]$. Suppose that 
\[\label{eq:adverse} \theta_L x^e_L - \psi_L(x^e_L)  < \theta_H x^e_H - \psi_L(x^e_H) \,\]
so the low type wants to imitate the high type when given the menu of the efficient allocations with competitive prices. Without this assumption, there is no adverse selection problem. Suppose also that there exists some $x \geq x^e_L$ such that $\theta_L x^e_L - \psi_L(x^e_L)  \geq \theta_H x - \psi_L(x)$ so it is possible to separate the types using only the work allocations. 

There is one costly instrument. For a level $y \in [0,1]$ of the costly activity, a type-$\theta_i$ worker incurs a cost $c_i(y)$ where $c_i$ is a strictly increasing, continuously differentiable function on $[0, 1]$ with $c_i(0) = 0$. Suppose $\label{eq:effective} c'_L(0) > \psi'_L(1)$ and $c'_H(0) = 0$. This says that a small amount of $y$ costs nothing for the high type but a lot for the low type. 

The firms commit to a set of offers. Each offer specifies an amount of work $x$, a level of costly activity $y$, and a wage $w$.  The literature has not reached a consensus on the choice of solution concept for competitive screening models. We say a set of offers $\{(x, y, w)\}$ is a \textit{\textbf{separating set}} if \textit{(i)} the types separate and \textit{(ii)} the firms earn zero payoff on each offer. A set of offers is a \textit{\textbf{Pareto-optimal separating set}} if it is (constrained) Pareto-optimal among all separating sets. This solution concept is weaker than the Pareto-dominant separating set, which is known to be equivalent to the reactive equilibrium of \citetapp{riley1979informational} in settings with one screening device (\citealtapp{engers1987market}).

This competitive screening model is analogous to the labor market application in \Cref{subsec:labor}. However, costly screening now emerges in equilibrium: 
\begin{prop}\label{prop:comp}
A Pareto-optimal separating set exists and any Pareto-optimal separating set involves costly screening. 
\end{prop}

\begin{proof}[Proof of \Cref{prop:comp}]
We first prove the second part. Suppose for contradiction that there exists a Pareto-optimal separating set $\{(x, y, w)\}$ that does not involve costly screening ($y = 0$). By the definition of a separating set, $x_H, x_L$ must differ. By the single-crossing property of $\psi$, we then have $x_H > x_L$. Note that $x_H$ cannot be $x_H^e$ because if so $\IC[\theta_L \rightarrow \theta_H]$ will be violated: 
\[\theta_L x_L - \psi_L(x_L)  \leq \theta_L x^e_L - \psi_L(x^e_L)  < \theta_H x^e_H - \psi_L(x^e_H)\,,\]
where the first inequality holds by definition of $x^e_L$ and the second inequality holds by assumption. Therefore, $\IC[\theta_L \rightarrow \theta_H]$ must be binding. To see this, note that if the upward IC constraint is not binding, then one can move $x_H$ by small enough $\delta$ toward $x^e_H$ without breaking the upward IC constraint. Since the surplus function $\theta_H x -\psi_H(x)$ is strictly concave, the modification increases the payoff of the high type and hence also preserves the downward IC constraint. But this means that the original set of offers is dominated by a separating set and hence impossible. 

Since $x_H > x_L$ and the upward IC constraint is binding, the downward IC constraint must be slack by the single-crossing property of $\psi$. This implies that $x_L = x^e_L$ because otherwise moving $x_L$ slightly toward $x^e_L$ gives a contradiction by the same argument as above. We claim that $x_H > x^e_H$. To see this, let 
\[f(x) = (\theta_H x - \psi_L(x)) - (\theta_L x^e_L - \psi_L(x^e_L)) \,. \]
Note that it is concave on $[x^e_L, x^e_H]$. Moreover, $f(x^e_L) = (\theta_H - \theta_L)x^e_L > 0$, and  $f(x^e_H) = (\theta_H x^e_H - \psi_L(x^e_H)) - (\theta_L x^e_L - \psi_L(x^e_L)) > 0$ by assumption. Thus, $f(x) > 0$ for all $x \in [x^e_L, x^e_H]$ and hence $x_H$ cannot be in that region. Therefore, $x_H > x^e_H$. 

Now consider the menu $\{(x_L, 0, \theta_L x_L), (x_H-\epsilon, \epsilon ,\theta_H (x_H - \epsilon))\}$ for $\epsilon > 0$. We claim that for $\epsilon$ small enough, the offer $(x_H-\epsilon, \epsilon ,\theta_H (x_H - \epsilon))$ increases the payoff of the high type. Let 
\[u_H(\epsilon) = \theta_H (x_H - \epsilon) -\psi_H(x_H - \epsilon) - c_H(\epsilon)\,.\]
It is a continuously differentiable function of $\epsilon$. The right derivative of this function at $0$ is strictly positive because 
\[\partial_{+} u_H(0) =  - (\theta_H  - \partial_{-}\psi_H(x_H)) - \partial_{+} c_H(0) =  - (\theta_H  - \partial_{-}\psi_H(x_H)) > 0\,,\]
where the second equality holds by assumption, and the last inequality holds by strict concavity of $\psi_H$ and that $x_H > x^e_H$. Therefore, there exists some $\epsilon > 0$ such that $u'_H(s) > 0$ for all $s \in [0, \epsilon]$; the claim follows immediately. 

We also claim that for $\epsilon > 0$ small enough, the modification still deters the low type from imitating the high type. To see this, let  
\[\hat{u}_L(\epsilon) = \theta_H (x_H - \epsilon) -\psi_L(x_H - \epsilon) - c_L(\epsilon) \,.\]
It is a continuously differentiable function of $\epsilon$. The right derivative of this function at $0$ is strictly negative because 
\[\partial_+ \hat{u}_L(0)= -\theta_H + \partial_{-} \psi_L(x_H) - \partial_+ c_L(0) \leq \partial_{-} \psi_L(1) - \partial_+ c_L(0) < 0 \,,\]
where the first inequality uses convexity of $\psi_L$ and the second inequality holds by assumption. Therefore, there exists some $\epsilon > 0$ such that $\hat{u}'_L(s) < 0$ for all $s \in [0, \epsilon]$; the claim follows immediately.  

Hence, for $\epsilon > 0$ sufficiently small, the proposed menu is a separating set that Pareto-improves on the original one. Contradiction. 

For the first part of the statement, consider the following optimization problem:
\[\max_{(x,y)\in[0,1]^2} \theta_H x -\psi_H(x) - c_H(y) \]
\[\text{subject to} \quad \theta_L x^e_L -\psi_L(x^e_L)  \geq \theta_H x -\psi_L(x) - c_L(y)\,. \]
An optimizer $(x^*, y^*)$ exists by standard compactness arguments. Moreover, $y^* \neq 0$ because otherwise it can be strictly improved by the argument above. 

We claim that $\{(x^e_L, 0, \theta_L x^e_L), (x^*, y^*, \theta_H x^*)\}$ is a separating set. The low type chooses the first offer by construction. To see that the high type chooses the second offer, recall that by assumption there exists some $x \geq x^e_L$ such that $\theta_L x^e_L - \psi_L(x^e_L)  \geq \theta_H x - \psi_L(x)$. Since this inequality is violated at $x^e_L$, by continuity, there exists some $x_H > x^e_L$ such that $\theta_L x^e_L - \psi_L(x^e_L)  = \theta_H x_H - \psi_L(x_H)$. Then, by the single-crossing property of $\psi$, $\theta_H x_H - \psi_H(x_H) > \theta_L x^e_L - \psi_H(x^e_L)$. Thus,
\[ \theta_H x^* -\psi_H(x^*) - c_H(y^*) \geq  \theta_H x_H - \psi_H(x_H) >  \theta_L x^e_L - \psi_H(x^e_L)\,, \]
where the first inequality uses that $(x_H, 0)$ is a feasible solution. So the high type chooses the second offer.   

Observe that $\{(x^e_L, 0, \theta_L x^e_L), (x^*, y^*, \theta_H x^*)\}$ must be Pareto-optimal among all separating sets. Suppose for contradiction that there is a separating set that Pareto-dominates it. Then the separating set must provide strictly higher payoff for the high type and maintain the same payoff for the low type because the low type already gets the maximal payoff. But that is impossible subject to $\IC[\theta_L \rightarrow \theta_H]$ by the construction of $(x^*, y^*)$. 
\end{proof}

\subsection{Additional Proofs}  \label{app:add}

\subsubsection{Measurable Monotone Coupling}
In this appendix, we provide the proof of \Cref{lem:decomp}, building on \citetapp{kamae1978stochastic}. We start with a technical lemma. 

\begin{lemma}\label{lem:eqv}
Suppose $X$, $Y$ are two $\Theta$-valued random variables where $\Theta$ is a compact subset of $\R^N$. Let $\preceq_{st}$ and $\preceq'_{st}$ denote the stochastic dominance partial orders on $\Delta(\R^N)$ and $\Delta(\Theta)$ (i.e. $X \preceq'_{st} Y$ if $\E[f(X)]\leq \E[f(Y)]$ for all bounded nondecreasing measurable $f: \Theta \rightarrow \R$). Then $X \preceq_{st} Y$ if and only if $X \preceq'_{st} Y$.
\end{lemma}

\begin{proof}[Proof of \Cref{lem:eqv}]
$(\impliedby)$ Suppose $X \preceq'_{st} Y$. Note that for any bounded monotone measurable $f: \R^N \rightarrow \R$, the restriction $f|_{\Theta}: \Theta \rightarrow \R$ is also a bounded monotone measurable function, and moreover
$\E[f(X)] = \E[f|_{\Theta}(X)]\leq \E[f|_{\Theta}(Y)] =  \E[f(Y)]$ since $X$, $Y$ are $\Theta$-valued and $X \preceq'_{st} Y$. So $X \preceq_{st} Y$.

$(\implies)$ Suppose $X \preceq_{st} Y$. To show $X \preceq'_{st} Y$, by  Theorem 1 of \citetapp{kamae1977stochastic}, it suffices to show that for any increasing set $B \subseteq \Theta$ closed in $\Theta$, we have $\E[\1_{X \in B}]\leq \E[\1_{Y \in B}]$ (we say a set $B$ is \textit{increasing} if $\1_{B}$ is a nondecreasing function).

Fix any such $B$. Let $B^{\uparrow}:= \{y \in \R^N: y \geq x, x \in B\}$
be the increasing hull of $B$ in $\R^N$. We claim that $B^{\uparrow}$ is closed in $\R^N$. To see this, fix any $y^n \rightarrow y$ in $\R^N$ where $y^n \in B^{\uparrow}$. Since $y^n \in B^{\uparrow}$, there exists $x^n \in B$ such that $y^n \geq x^n$. Since $B$ is a closed subset of a compact set $\Theta$, $B$ is compact. Therefore, there exists a subsequence $x^{n_l}$ converging to some $x \in B$. Passing to this subsequence, we have $\displaystyle y = \lim_{l\rightarrow \infty} y^{n_l} \geq \lim_{l \rightarrow \infty} x^{n_l} = x \in B$ and hence $y \in B^{\uparrow}$. This proves that $B^{\uparrow}$ is closed in $\R^N$, and hence measurable. 

Because $X$ is $\Theta$-valued, we have $\E[\1_{X\in B^{\uparrow}}] = \E[\1_{X\in B^{\uparrow} \cap \Theta }]$. We claim that $ B^{\uparrow} \cap \Theta = B$. Since $B \subseteq \Theta$ and $B \subseteq  B^{\uparrow}$, we have $ B \subseteq B^{\uparrow} \cap \Theta$. Now take any $y\in B^{\uparrow} \cap \Theta$. Then $y \in \Theta$ and there exists some $x \in B$ such that $y \geq x$. But because $B$ itself is an increasing set in $\Theta$, we must have $y \in B$. Thus $B^{\uparrow} \cap \Theta \subseteq B$. Therefore, $B^{\uparrow} \cap \Theta = B$. Now, we have 
\[\E[\1_{X\in B}]  =  \E[\1_{X\in B^{\uparrow} \cap \Theta }] = \E[\1_{X\in B^{\uparrow}}] \leq \E[\1_{Y\in B^{\uparrow}}] = \E[\1_{Y\in B^{\uparrow} \cap \Theta }] = \E[\1_{Y\in B}]  \]
where the inequality follows from that $X \preceq_{st} Y$ and $B^{\uparrow}$ is a measurable increasing set in $\R^N$. Since this holds for all closed increasing sets $B$ in $\Theta$, the claim follows.
\end{proof}

\vspace{0.5cm}

\begin{proof}[Proof of \Cref{lem:decomp}]

Let $\mathcal{B}_{\Theta^B}$ be the Borel $\sigma$-algebra of $\Theta^B$. Let $\kappa: \Theta^A \times \mathcal{B}_{\Theta^B} \rightarrow [0, 1]$ be the regular conditional distribution of $\theta^B$ given $\theta^A$. For any $t \in \Theta^A$, define the measure $P_t(\,\cdot\,) = \kappa(t, \,\cdot\,)$. Let $S$ be the support of $\theta^A$. By assumption, $\{P_t\}_{t \in S}$ is a stochastically nondecreasing family of probability measures on $\R^N$. By \Cref{lem:eqv}, it is also a stochastically nondecreasing family of probability measures on $\Theta^B$. 

Let 
\[\varphi(s) = \begin{cases}
\max \{t:\, t \leq s, \, t \in S\} &\text{ if $s \geq \min(S)$} \\
\min(S) &\text{otherwise} \\
\end{cases} \,.\]
Because $(-\infty, s] \cap S$ is compact, we have $\varphi(s) \in S$. For all $s \not \in S$, define 
$P_s = P_{\varphi(s)}$.
Because $\varphi(\,\cdot\,)$ is nondecreasing and $\varphi(s) = s$ for all $s \in S$, $\{P_t\}_{t \in \R}$ is a stochastically nondecreasing family of probability measures on $\Theta^B$. Invoking Theorem 6 of \citetapp{kamae1978stochastic}, we have a $\Theta^B$-valued stochastic process $\{X_t\}_{t\in \R}$ on a probability space $(\Omega, \nu)$ such that (i) $X_s(\omega) \leq X_t(\omega)$ for all $s < t$ and all $\omega$ and (ii) $P_t$ is the distribution of $X_t$ for all $t$. 

By the proof of Theorem 6 in \citetapp{kamae1978stochastic}, there exists a dense countable set $D \subset \R$ such that for all $\omega$ and all $s \not \in D$
\[X_s(\omega) = \lim_{t \rightarrow s; \, t \in D, \,t \leq s} X_t(\omega)\,. \]
Let $X^i_t$ denote the $i$-th coordinate of $X_t$. We claim that for all $i$ and all $\omega$, the sample path $X^i_t(\omega)$ is left-continuous at all $t \not \in D$. To see this, fix any $i$, $\omega \in \Omega$, $t \not \in D$, and $\epsilon > 0$. By construction, $X^i_{t}(\omega) = \displaystyle \lim_{k \rightarrow \infty} X^i_{t_k}(\omega)$ for some sequence $t_k \uparrow t$, with $t_k \in D$. So there exists some $K \in \N$ such that $X^i_t(\omega) - X^i_{t_{K}}(\omega)<\epsilon$. But then for any $s \in (t - \delta, t)$ where $\delta := t - t_{K}$, we have $|X^i_t(\omega) - X^i_s(\omega)| \leq X^i_t(\omega) - X^i_{t_{K}}(\omega) < \epsilon$ by monotonicity of $X^i_t(\omega)$. 

Because $D$ is countable, $D^c$ is dense in $\R$. Pick any dense countable set $Q \subset D^c$. For all $\omega$, define $\bar{X}_t(\omega) = X_t(\omega)$ for all $t \in D^c$ and  
\[\bar{X}_t(\omega) =  \lim_{s \rightarrow t; \, s \in Q, \, s \leq t} \bar{X}_s(\omega)\]
for all $t \in D$. Note that $\{\bar{X}_t\}_{t\in\R}$ is also a nondecreasing stochastic process. By a similar argument as above, $\bar{X}^i_t(\omega)$ is left-continuous at all $t \in D$. Moreover, for any $t \in D^c$, and any sequence $t_k \uparrow t$, with $t_k \in Q$, we have  $\bar{X}^i_t(\omega) = X^i_t(\omega) = \displaystyle \lim_{k \rightarrow \infty} X^i_{t_k}(\omega) = \displaystyle \lim_{k \rightarrow \infty} \bar{X}^i_{t_k}(\omega)$ by left continuity of $X^i_t(\omega)$ at $t \in D^c$. Therefore, by a similar argument as above, $\bar{X}^i_t(\omega)$ is also left-continuous at all $t \in D^c$. So $\{\bar{X}_t\}_{t\in\R}$ is a left-continuous stochastic process, and thus the map $(t, \omega) \mapsto \bar{X}^i_t(\omega)$ is jointly measurable (see e.g. \citealtapp{Karatzas1998}, p. 5). Since $D$ is countable, $(t, \omega) \mapsto X^i_t(\omega)\1_{t \in D}$ is also jointly measurable. Then, because $X^i_t(\omega) = \bar{X}^i_t(\omega)\1_{t \not \in D} +  X^i_t(\omega)\1_{t \in D}$, $(t, \omega) \mapsto X^i_t(\omega)$ is jointly measurable. Therefore, $(t, \omega) \mapsto X_t(\omega)$ is jointly measurable. 

Let $\mathcal{E} = \Omega$ and $h(\theta^A; \varepsilon) = X_{\theta^A}(\varepsilon)$ for all $\theta^A \in \Theta^A$ and $\varepsilon \in \mathcal{E}$. Then $h:\Theta^A \times \mathcal{E} \rightarrow \Theta^B$ is jointly measurable and nondecreasing in the first argument. Let $\mu^A$ denote the marginal distribution of $\theta^A$. By the construction of $h$, for any $(a, b) \in \R \times \R^N$, we have 
\begin{align*}
    (\mu^A \times \nu)(\{(\theta^A, \varepsilon): \theta^A \leq a,\, h(\theta^A; \varepsilon) \leq b\}) &= \int \1_{\theta^A \leq a} \1_{\theta^A \in S} \Big(\int\1_{h(\theta^A; \varepsilon)\leq b} \d\nu(\varepsilon) \Big)\d\mu^A(\theta^A)\\
    &=  \int \1_{\theta^A \leq a} \1_{\theta^A \in S} \kappa (\theta^A, \,\{\theta^B: \theta^B \leq b\} )\d\mu^A(\theta^A)\\
    &= \int \1_{\theta^A \leq a} \kappa(\theta^A, \,\{\theta^B: \theta^B \leq b\}) \d\mu^A(\theta^A) \\
    & = \P(\theta^A \leq a, \,\theta^B \leq b)\,,
\end{align*}
where we have also used  $\mu^A(S) = 1$ and Fubini's theorem. Thus, $\theta \eqid (\theta^A, h(\theta^A, \varepsilon))$ when $\theta^A, \, \varepsilon$ are independently drawn from $\mu^A, \, \nu$ respectively. 
\end{proof}

\subsubsection{Stochastic Monotonicity and Affiliation}
In this appendix, we show that affiliation implies stochastic monotonicity (this is standard; we include it here for completeness).

\begin{lemma}\label{lem:aff}
Suppose $(\theta^A, \theta^B)$ is a jointly affiliated random variable (where $\theta^A$ is $\R$-valued and $\theta^B$ is $\R^N$-valued). Then, $\theta^B$ is stochastically nondecreasing in $\theta^A$. 
\end{lemma}
\begin{proof}[Proof of \Cref{lem:aff}]
By Theorem 3.10.16 in  \citetapp{muller2002comparison}, affiliation implies the ``conditionally increasing'' property (CI) defined in Definition 3.10.9 in \citetapp{muller2002comparison}. Moreover, by Theorem 5.3 of \citetapp*{block1985concept}, CI implies the ``positively dependent through stochastic ordering'' property (PDS) which is also defined in Definition 3.10.9 in \citetapp{muller2002comparison}. Note that $(\theta^A, \theta^B)$ satisfying PDS clearly implies that $\theta^B$ is stochastically nondecreasing in $\theta^A$, proving the claim. 
\end{proof}

\setlength\bibsep{10pt}
\bibliographystyleapp{ecta} 
\bibliographyapp{references}
\end{document}